\documentclass[11pt,reqno]{amsart}
\usepackage[USenglish]{babel}
\usepackage{a4wide} 
\usepackage[numbers]{natbib}
\usepackage{amsthm}
\usepackage{amssymb} 
\usepackage{amsmath}
\usepackage{amsfonts}
\usepackage{verbatim}
\usepackage{centernot} 
\usepackage{stmaryrd} 
\usepackage{mathrsfs} 
\usepackage{paralist} 
\usepackage{tikz}
\usetikzlibrary{arrows,shapes}
\usepackage{enumerate}
\usepackage[matrix,arrow,curve]{xy}
\usepackage[parfill]{parskip}
\usepackage{pifont}    

\title{Stochastic Nondeterminism and Effectivity Functions}
\address{Chair for Software Technology\\Technische Universit\"at Dortmund\\Joseph-von-Fraunhofer-Str.\ 23\\D-44227. Germany}
\email{ernst-erich.doberkat@udo.edu}
\curraddr{Math ++ Software, Lewackerstr. 6 b, D-44879 Bochum, Germany}
\email{eed@doberkat.de}
\address{CIEM --- Facultad de Matem\'atica, Astronom\'{\i}a y F\'{\i}sica 
    (Fa.M.A.F.)\\Universidad Nacional de C\'ordoba --- Ciudad Universitaria\\C\'ordoba 5000. Argentina.}
\email{sterraf@famaf.unc.edu.ar}
\author{Ernst-Erich Doberkat \quad\quad Pedro S\'{a}nchez Terraf}
\thanks{This article has been accepted for publication in the Journal
  of Logic and Computation Published by Oxford University Press.\\ 
  The
    second author was partially
    supported by CONICET, ANPCyT project PICT 2012-1823, SeCyT-UNC project 05/B284, and by DFG grant DO 263/13-1.}
\def\cal{\mathcal}

\def\paragraph{\subsubsection*}

\newcommand{\Inv}[2]{\ensuremath{{\mathcal INV}\left(#1, #2\right)}}
\newcommand{\Klasse}[2]{\left[#1\right]_{#2}}
\newcommand{\Faktor}[2]{{#1}/{#2}}
\newcommand{\fMap}[1]{\eta_{#1}}
\newcommand{\Bild}[2]{{#1}\left[#2\right]}
\newcommand{\InvBild}[2]{\Bild{#1^{-1}}{#2}}
\newcommand{\Kern}[1]{\mathsf{ker}\left(#1\right)}
\newcommand{\Folge}[1]{(#1_n)_{n \in \Nat}}

\newcommand{\SubProb}[1]{\spaceFont{S}\left(#1\right)}

\newcommand{\SubProbSenza}{\spaceFont{S}}

\newcommand{\PowerSet}[1]{\ensuremath{\mathcal{P}{\left(#1\right)}}}
\newcommand{\PowerSenza}{\ensuremath{\mathcal{P}}}

\newcommand{\Borel}[1]{\ensuremath{{\mathcal B}(#1)}}

\newcommand{\Trans}{\rightsquigarrow}

\newtheorem{definition}{Definition}[section]
\newcommand{\BeginDefinition}[1]{
  \begin{definition}\label{#1}
}
\newcommand{\EndDefinition}{\end{definition}}

\newtheorem{example}[definition]{Example}
\newcommand{\BeginExample}[1]{
  \begin{example}\label{#1}\rm
}
\newcommand{\EndExample}{--- \end{example}}

\newtheorem{observation}[definition]{Observation}
\newcommand{\BeginObservation}[1]{
  \begin{observation}\label{#1}\rm
}
\newcommand{\EndObservation}{--- \end{observation}}

\newtheorem{theorem}[definition]{Theorem}
\newcommand{\BeginTheorem}[1]{
  \begin{theorem}\label{#1}
}
\newcommand{\EndTheorem}{\end{theorem}}

\newtheorem{corollary}[definition]{Corollary}
\newcommand{\BeginCorollary}[1]{
  \begin{corollary}\label{#1}
}

\newtheorem{proposition}[definition]{Proposition}
\newcommand{\BeginProposition}[1]{
  \begin{proposition}\label{#1}
}
\newcommand{\EndProposition}{\end{proposition}}

\newcommand{\EndCorollary}{\end{corollary}}
\newtheorem{lemma}[definition]{Lemma}
\newcommand{\BeginLemma}[1]{
  \begin{lemma}\label{#1}
}
\newcommand{\EndLemma}{\end{lemma}}

\newtheorem{claim}{Claim}
\newcommand{\BeginClaim}[1]{
  \begin{claim}\label{#1}
}
\newcommand{\EndClaim}{\end{claim}}

\newcommand{\BeginProof}{\begin{proof}}
\newcommand{\EndProof}{\end{proof}}

\newenvironment{remark}{\textbf{Remark:\ }}{}
\newcommand{\BeginRemark}{\begin{remark}}
\newcommand{\EndRemark}{\QED\end{remark}}
\newcommand{\QED}{
\ensuremath{\dashv}
}

\newcommand{\Real}{\mathbb{R}}
\newcommand{\pReal}{\mathbb{R}_{+}}
\newcommand{\Nat}{\mathbb{N}}
\newcommand{\Rational}{\mathbb{Q}}

\newcommand{\rel}[1]{\mathrel{\mathcal{R}{\left(#1\right)}}}
\newcommand{\sem}[1]{\llbracket #1 \rrbracket}

\newcommand{\bisim}{\sim}
\newcommand{\sbisim}{\bisim_{\rm s}}
\newcommand{\ebisim}{\bisim_{\rm e}}

\newcommand{\nlmp}[1]{\mathbf{#1}}

\renewcommand{\emptyset}{\varnothing}

\newcommand{\et}{\mathrel{\&}}
\newcommand{\dia}{\lozenge}
\renewcommand{\Box}{\mathop{\oblong\mspace{-1mu}}}
\newcommand{\calS}{\mathcal{S}}
\newcommand{\calU}{\Xi}
\renewcommand{\calU}{{\cal D}}

\newcommand{\Pow}{\PowerSenza}

\newcommand{\comp}{\mathsf{c}}

\newcommand{\B}{\Borel{}}

\newcommand{\N}{\mathbb{N}}

\newcommand{\<}{\langle}
\renewcommand{\>}{\rangle}

\renewcommand{\keywords}[1]{{\renewcommand{\thefootnote}{\relax}\footnotetext{\emph{Keywords:}
    #1}}}

\newcommand{\ent}{\Rightarrow}
\newcommand{\tne}{\Leftarrow}
\newcommand{\sii}{\Leftrightarrow}
\renewcommand{\phi}{\varphi}

\newcommand{\Lda}{\Lambda}
\renewcommand{\Lda}{{\cal A}} 
\newcommand{\lda}{\lambda}
\newcommand{\sig}{\ensuremath{\sigma}}

\newcommand{\fil}{\mathop{\mathfrak{F}}}
\renewcommand{\fil}{\mathfrak{f}} 
\renewcommand{\fil}[1]{@ #1} 
\renewcommand{\fil}[1]{\ensuremath{\lfloor{#1}\rfloor}}

\newcommand{\st}{\mid}

\definecolor{lightblue}{RGB}{224,224,255}
\definecolor{lightred}{RGB}{255,224,224}
\definecolor{lightgreen}{RGB}{224,255,224}
\definecolor{lightyellow}{RGB}{255,255,224}
\definecolor{lightpurple}{RGB}{255,224,255}
\definecolor{darkerred}{RGB}{64,0,0}
\definecolor{darkred}{RGB}{128,0,0}
\definecolor{darkblue}{RGB}{0,0,128}
\definecolor{darkgreen}{RGB}{0,128,0}
\definecolor{darkpurple}{RGB}{128,0,128}

\newcommand{\colorpar}[3]{\colorbox{#1}{\parbox{#2}{#3}}}

\newcommand{\marginremark}[4]{\marginpar{\colorpar{#2}{\linewidth}{\color{#1}\tiny{[#3]~ #4}}}}

\makeatletter
\def\THICKhrulefill{\leavevmode \leaders \hrule height 5pt\hfill \kern \z@}
\makeatother
\newcommand{\highlightedremark}[4]{\begin{center}\fcolorbox{#1}{#2}{\begin{minipage}{.98\linewidth}\color{#1}\textbf{\THICKhrulefill[ #3 ]\THICKhrulefill}\par\noindent#4\end{minipage}}\end{center}}
\makeatletter
\def\@remk#1#2#3#4{\marginremark{#3}{#4}{#1}{#2}}
\def\remarkPST#1{\@remk{PST}{#1}{darkred}{lightred}}
\def\remarkEED#1{\@remk{EED}{#1}{darkblue}{lightblue}}
\def\@hrmk#1#2#3#4{\highlightedremark{#3}{#4}{#1}{#2}}
\def\hrmkEED#1{\@hrmk{EED}{#1}{darkgreen}{lightblue}}
\def\hrmkPST#1{\@hrmk{PST}{#1}{darkpurple}{lightred}}
\makeatother

\newtheorem{prop}[definition]{Proposition}
\theoremstyle{remark}
\newtheorem*{ack}{Acknowlegdements}

\newcommand{\sigalg}{\mathscr{S}} 
\renewcommand{\B}{\mathcal{B}}          
\newcommand{\subp}{\mathbb{S}}          
\newcommand{\hit}{\mathscr{H}}
\newcommand{\hitt}{\mathsf{H}}
\newcommand{\cones}{\VauSenza}         
\newcommand{\ka}{\kappa}
\newcommand{\EF}{\mathscr{E\negthickspace{F}}\!}
\newcommand{\NK}{\mathscr{N\negthickspace{K}}\!}          
\newcommand{\two}{{2\mathcal{L}}}       
\newcommand{\eqlog}{\approx}
\newcommand{\weak}{w}
 
\newcommand{\calC}{\mathcal{C}} 
\newcommand{\bR}{\mathrel{\bar R}}

\newcommand{\iotac}{\iota_\calC}
\renewcommand{\xi}{X}                   

\def\QED{\qed}
\newcommand{\isEquiv}[3]{\ensuremath{{#1}\ {#3}\ {#2}}}
\newcommand{\basS}[3]{\bas{#2}{#3 \ast #1}}
\renewcommand{\basS}[3]{\bas{#2}{#3\,\pmb{\mid}\,#1}}
\renewcommand{\basS}[3]{{\boldsymbol{\beta}}_{#1}(#2, #3)}
\renewcommand{\Inv}[2]{\ensuremath{{\Sigma}(#1, #2)}}
\renewcommand{\Inv}[2]{\ensuremath{{\Sigma}_{#1}(#2)}}
\renewcommand{\SubProbSenza}{\mathbb{S}}
\renewcommand{\SubProb}[1]{\SubProbSenza(#1)}
\newcommand{\bSS}[1]{{#1}\otimes\Borel{[0, 1]}}
\def\funktorFont{\mathbb}
\def\VauSenza{\ensuremath{\funktorFont{V}}}
\newcommand{\Vau}[1]{\ensuremath{\ensuremath{\VauSenza(#1)}}}
\def\eTrans{\pmb{\twoheadrightarrow}}
\def\doteq{:=}
\def\Xf{\ensuremath{\Sigma_{f}}}
\def\ul#1{#1}
\def\MMP#1{}

\begin{document}

\begin{abstract}
This paper investigates stochastic nondeterminism by relating
nondeterministic labelled Markov processes and stochastic effectivity
functions to each other; the underlying state spaces are
continuous. Both generalizations to labelled Markov transition systems
have been proposed recently with differing intentions. It turns out
that they display surprising similarities and interesting differences,
as we will demonstrate in this paper.
\end{abstract}

\maketitle
\section{Introduction}

Nondeterministic labelled Markov processes (henceforth abbreviated as
NLMPs) propose a model for stochastic non-determinism in which a state
is assigned a measurable set of probability distributions as a way of
selecting non-deterministically the distribution of the next state in
a transition system. As outlined in~\cite{Terraf-Bisim-MSCS,
  DWTC09:qest}, this serves as a model for internal non-determinism; a
crucial point here is the question of measurability for the underlying
transition law. To give the idea, at the heart of a NLMP over a
measurable state space $S$ lies a family of functions $(\ka_{a})_{a\in
  A}$ indexed by actions from a set $A$, each of which assigns to a
state $s$ a set $\ka_{a}(s)$ of probabilities over $S$, modeling the
set of distributions which are possible after action $a$ in state
$s$. To obtain meaningful insights into the behavior of such a system,
some assumptions on measurability of $s \mapsto \ka_{a}(s)$ should be
imposed. The measurable structure is given by the hit
$\sigma$-algebra, a construction very similar to the hyperspace
constructions in topology~\cite{Michael, Kechris}. This construction
forms the basis for investigations into system behavior, most notably
into different variants of bisimulations. For example, it was
shown that the negation-free logic proposed in~\cite{Parma+Segala} is
a suitable logic for investigating event
bisimulation~\cite[Theorem~4.5]{Terraf-Bisim-MSCS}. This demonstrates
that NLMPs are an adequate tool for representing and investigating
stochastic non-determinism.

In a parallel development, stochastic effectivity functions have been
proposed for providing a stochastic interpretation of game logic,
continuing the line of research which has been initiated by
Parikh~\cite{Parikh-Games1985} and later continued
in~\cite{Pauly-Parikh} one one hand, and the stochastic interpretation
of propositional dynamic logic, a fragment of game logic, through
Kripke models on the other hand~\cite{EED-PDL-TR}. Parikh had observed
that neighborhood models are a suitable tool for the variant of modal
logics which he had proposed as game logic. A neighborhood model over
a set $W$ of worlds is essentially represented by an effectivity
function over $W$, i.e., a function $F$ which maps $W$ to the set of
all upper closed subsets of the power set $\PowerSet{W}$, so that for
$F(w)\subseteq \PowerSet{W}$ has the property that $A\in F(w)$
and $A\subseteq A'$ implies $A'\in F(w)$ all $w\in W$. Assume that
$w\in W$ represents the state of some game. Let $A\in F(w)$
indicate that, when making a move, the player has a strategy to reach
a state in $A$. Clearly $F(w)$ is upper closed. Now
assume that these games can be
combined in various ways, e.g., by composing them sequentially, by
choosing one rather than the other, or by iterating them a finite but
indefinite number of times. Having only one player is probably not too
entertaining, so it is assumed that the game has two players, which
are henceforth called \emph{Angel} and \emph{Demon}, and which are
assumed to take turns.  It is assumed that the entire game is
\emph{determined}, indicating that exactly one of both players has a
winning strategy~\cite[Chapter 33]{Jech}; for symmetry, each game has
a dual game which assumes that Angel and Demon change
r\^oles. Determinacy then implies that if Angel does not have a
winning strategy for a game, Demon has one for its dual, which means
that if Angel does not have a strategy for achieving a state in a set
$A$, Demon has a strategy for achieving a state in $W\setminus A$, and
vice versa. This implies that we have only to cater for Angel's
movements; assume the latter's effectivity function for a game is
given by $F$, then $w \mapsto\{W\setminus A \mid A\not\in F(w)\}$ will
take care of Demon's movements for this game. This scenario will
motivate part of our approach, so we will return to it from time to
time.

A probabilistic interpretation of game logic assumes that the outcome
of a game is modeled in terms of probability distributions over the
set $W$ of states or worlds, which means that we assign to each world
$w$ an upward closed set $F(w)$ of probability measures on $W$ for
modeling a game, and $A\in F(w)$ is intended to represent the
assumption that $A$ is a possible set of distributions which Angel
can achieve upon playing that game. But since we are in the realm of
probabilities, we have to make sure that the probabilities are defined
at all, so we have to assume also here the set of worlds $W$ being
endowed with a measurable structure, each element of $F(w)$ being a
measurable set of probabilities as well. So we set up the scenario of
assigning to each world $w$ a set $F(w)$ of measurable sets of
probabilities. With this basic assumptions we are able to model some,
but not all, constructions in games logic, the important operation of
composing games is not among them, however. For modeling this, we want
to be able to compose stochastic effectivity functions so that the
result can be used for representing sequential compositions; this
operation is technically a bit involved, as it requires an additional
property on these functions which we call
\emph{t-measurability}~\cite{EED-GameLogic-TR,EED-Alg-Prop_effFncts}. It
states roughly that we obtain a measurable function when taking
quantitative measurements into account; this will then enable us to
compose effectivity functions through integration.  It turns out that
t-measurability can be characterized topologically, that it is a
powerful concept for the purposes indicated here, but that it marks on
the other hand one of the boundaries between NLMPs and effectivity
functions.

Both NLMPs and stochastic effectivity functions are generalizations of
Markov kernels, a.k.a.\ stochastic relations.  It is not so
straightforward, however, to see a relationship between NLMPs and
effectivity functions, and here things start to become interesting. We
investigate the interplay between both by looking into conditions
under which they generate each other. For example, we devise a
mechanism which renders an effectivity function from an NLMP --- most
of the time, but when it appears to be most obvious, it does not
work. We show that in the finitely supported case there are some very
close connections between them in Polish spaces; this is established
through a selection argument. We review bisimulations for NLMPs
and translate these results to effectivity functions. On the other
hand, morphisms for effectivity functions give rise to an
investigation of morphisms for NLMPs.  We show that porting a concept
from one side to the other furthers insight into both ways of modeling
stochastic nondeterminism.

\paragraph{Outline} 
In the next section, we recall some basic material concerning
measurable spaces and subprobability measures. Highlighted concepts
are the \emph{hit $\sigma$-algebra} and the space of upper closed
sets. In Section~\ref{sec:main-char-nond-kern} we present the two
models of stochastic behavior which are studied in this work, namely,
nondeterministic kernels and effectivity functions; they have been
used to construct \emph{nondeterministic
  Labelled Markov Processes} and \emph{stochastic game models}. 

The investigation of  bisimulation is taken up in
Section~\ref{sec:bisimulations}, and serves as a framework for investigating
several issues. In Section~\ref{sec:representability}
the relation between non-determinism and effectivity is studied from a
measurability point of view, and also the interaction of
morphisms and relational bisimulations. A
logic for characterizing bisimilarity on finitary effectivity
functions is proposed in Section \ref{sec:two-level-logic}. Finally, a
coalgebraic perspective on bisimulation and behavioral equivalence is
developed in Section~\ref{sec:coalgebraic-approach}, where the notion
of \emph{subsystem} is developed and plays a central r\^ole.

Section~\ref{sec:conclusions--further-work} contains some further
directions and offers concluding remarks.

\section{$\sigma$-Algebras and All That}
\label{sec:sigm-all-that}
The reader is briefly reminded of some notions and constructions from measure theory,
including the famous $\pi$-$\lambda$-Theorem and Choquet's
representation of integrals as areas; these two tools are used all
over the paper. We also introduce invariant sets for an equivalence
relation together with the corresponding $\sigma$-algebra.

\subsection{Measurability}
First we fix some notations. A measurable space $(S, {\cal S})$ is a set $S$ with a
$\sigma$-algebra ${\cal S}$, i.e., ${\cal S}\subseteq\PowerSet{S}$ is
a Boolean algebra which is closed under countable unions. Here
$\PowerSet{S}$ is the power set of $S$. Given ${\cal S}_{0}\subseteq\PowerSet{S}$,
denote by 
\begin{equation*}
  \sigma({\cal S}_{0}) := \bigcap\{{\cal T} \mid {\cal
    S}_{0}\subseteq{\cal T}, {\cal T}\text{ is a $\sigma$-algebra}\}
\end{equation*}
the smallest $\sigma$-algebra containing ${\cal S}_{0}$ (the set for
which the intersection is constructed is not empty, since it contains
$\PowerSet{S}$). If $S$ is a topological space with topology $\tau$,
then the elements of $\sigma(\tau)$ are called the \emph{Borel sets}
of $S$; the $\sigma$-algebra $\sigma(\tau)$ is usually denoted by
$\Borel{S}$. For a measurable space  $(S, {\cal S})$  and $A\subseteq
S$ we define ${\cal S}\cap A := \{B\cap A \mid B\in{\cal S}\}$ as the
\emph{trace} of ${\cal S}$ on $A$, so that $(A, {\cal S}\cap A)$ becomes a
measurable space in its own right; note that we do not require $A$ to be a
measurable set. 

Given two measurable spaces $(S, {\cal S})$ and $(T,
{\cal T})$, the product space $(S\times T, {\cal S}\otimes{\cal T})$
has the Cartesian product $S\times T$ as a carrier set, the product
$\sigma$-algebra 
\begin{equation*}
  {\cal S}\otimes{\cal T} := \sigma(\{A\times B \mid A\in{\cal S}, B\in{\cal T}\})
\end{equation*}
is the smallest $\sigma$-algebra on $S\times T$ which contains all
rectangles $A\times B$ with $A\in {\cal S}$ and $B\in{\cal
  T}$. Define for $E\subseteq S\times T$ 
\begin{align*}
  E^{s} & := \{t\in T\mid \langle s, t\rangle\in E\}&&\text{ vertical
    cut,}\\
E_{t} & := \{s\in S \mid \langle s, t\rangle \in E\}&&\text{ horizontal
    cut\footnotemark,}
\end{align*}
\footnotetext{The
  notation of indicating the horizontal cut through an index conflicts
  with indexing, but it is customary, so we will be careful to make
  sure which meaning we have in mind.}
then $E^{s}\in{\cal T}$ for all $s\in S$, and $E_{t}\in{\cal S}$ for
all $t\in T$, provided $E\in{\cal S}\otimes{\cal T}$. The converse
does not hold, that is, a set having all of its cuts measurable is not
necessarily measurable~\cite[Exercise 21.20]{Hewitt-Stromberg}.

Dually, the \emph{coproduct} $(S, {\cal S})\oplus (T, {\cal T})$ of the
measurable spaces has as a carrier set the disjoint union $S\uplus T$ of the
carrier sets $S$ and $T$, and as a $\sigma$-algebra 
\begin{equation*}
  {\cal S}\oplus{\cal T} := \{C\subseteq S\uplus T \mid C\cap S\in
  {\cal S}\text{ and }C\cap T\in{\cal T}\}.
\end{equation*}
The coproduct of measurable spaces is sometimes also called their \emph{sum}. 

If $(S, \tau)$ and $(T, \theta)$ are topological spaces, then the
Borel sets $\Borel{\tau\times\theta}$ of the product topology may
properly contain the product $\Borel{\tau}\otimes\Borel{\theta}$. If,
however, both spaces are Hausdorff and $\theta$ has a countable basis,
then $\Borel{\tau\times\theta} =
\Borel{\tau}\otimes\Borel{\theta}$~\cite[Lemma 6.4.2]{Bogachev}. In
particular, the Borel sets of the product of two Polish spaces are
generated by Cartesian products of Borel sets from the components (a
\emph{Polish space} is a topological space which has a countable base
and for which a complete metric exists). The same applies to analytic
spaces (an \emph{analytic space} is a separable metric space which is
the image of a continuous map between Polish spaces). This is so
because the topology of these spaces is also countably
generated~\cite[Corollary 2.97]{EED-Meas}. A \emph{standard Borel
  space} is a measurable space the $\sigma$-algebra of which is
generated by a Polish topology.

In summary, the observation on products mentioned above suggests that
we have to exercise particular care when working with the product of
measurable spaces, which carry a topological structure as well.

Given the measurable spaces $(S, {\cal S})$ and $(T, {\cal T})$, a map
$f: S \to T$ is said to be \emph{${\cal S}$-${\cal T}$-measurable} iff
$\InvBild{f}{D}\in{\cal S}$ for all $D\in {\cal T}$. Call the
measurable map $f: (S, {\cal S}) \to (T, {\cal T})$ \emph{final} iff
${\cal T}$ is the largest $\sigma$-algebra $\mathcal{C}$ on $T$ such
that $\InvBild{f}{\mathcal{C}} := \{\InvBild{f}{C}\mid C\in {\cal
  C}\}\subseteq {\cal S}$ holds, so that $ {\cal T}= \{B \subseteq T
\mid \InvBild{f}{B}\in{\cal S}\}. $ Hence we may conclude from
$\InvBild{f}{B}\in{\cal S}$ that $B \in {\cal T}$, if $f$ is onto.  An
equivalent formulation for finality of $f$ is that a map $g: T \to U$
is ${\cal T}$-${\cal U}$-measurable if and only if $g\circ f: S\to U$
is ${\cal S}$-${\cal U}$-measurable, whenever $(U, {\cal U})$ is a
measurable space. Measurability of real valued maps always refers to
the Borel sets on the reals, hence $f: S\to \Real$ is measurable iff
$\{s\in S\mid f(s)\bowtie q\}\in{\cal S}$ for each rational number
$q$, with $\bowtie$ as one of the relations $\leq, <, \geq, >$.

{
\def\Lda{{\cal A}}
\def\xi{C}
\def\zeta{D}
\BeginDefinition{def:hit-salg}

  Let $\Lda$ be some family of sets. The
  \emph{hit \sig-algebra} $\hit(\Lda)$ is the least
  $\sigma$-algebra on $\Lda$ containing all sets
  $\hitt_{\xi} := 
    \{ \zeta\in\Lda :
                   \zeta\cap\xi \neq\emptyset \}$
  with $\xi\in\Lda$.
\EndDefinition
This $\sigma$-algebra will be used when we formulate the measurable
structure underlying nondeterministic labelled Markov processes.
This is an easy criterion for hit-measurability.

\BeginLemma{l:H-meas_subseteq}
Let $(S, {\cal S})$ be a measurable space, $\Lda$ a \sig-algebra on a set $T$, and a map $\ka:S\to
\Lda$. Then $\ka$ is ${\cal S}$-$\hit(\Lda)$ measurable iff for every
$\xi\in\Lda$, $\{s\in S : \ka(s)\subseteq \xi\} \in {\cal S}$.
\EndLemma

\begin{proof}
Just note that $\{s\in S : \ka(s)\subseteq \xi\} = S\setminus\{s\in S : \ka(s)\cap
T\setminus\xi \neq \emptyset\}$.
\end{proof}
}
\medskip

Some notation concerning binary relations will be needed. Let $R$ a
binary relation over $S$.  A set $Q$ is \emph{$R$-closed} if $x\in Q$
and $x\mathrel{R} y$ imply $y\in Q$; this is the appropriate
generalization of invariance for equivalence
relations. $\Inv{R}{{\cal S}}$, $\Inv{R}{S}$ or $\Sigma_R$ will
denote the \sig-algebra of $R$-closed sets in ${\cal S}$. If $\mu,
\mu'$ are measures defined on $\calS$, we write $\mu \bR\mu'$ if they
coincide over $\Inv{R}{S}$. Lastly, let $\calU$ be a subset of
$\Pow(S)$, the powerset of $S$. The relation $\rel{\calU}$ is given
by:
\begin{equation}\label{eq:rel-from-sig-algebra}
\langle s,t\rangle\in{\rel{\calU}} \quad \iff \quad \forall Q\in \calU: s\in Q \sii
t\in Q.
\end{equation}
Then the following observation, which is proved in~\cite[Lemma 3.1.6]{Srivastava}, is sometimes helpful

\BeginLemma{equiv-generated}
$\rel{\calU} = \rel{\sigma(\calU)}.$
\QED
\EndLemma

If $R$ is an equivalence relation, then an $R$-closed set is the union
of equivalence classes; we call an $R$-closed set in this case
\emph{$R$-invariant}.  As usual,
\begin{equation*}
\Kern{f} := \{\langle s, s'\rangle \mid f(s) = f(s')\}  
\end{equation*}
is the \emph{kernel of $f$}.

We know that in general the image of a Borel set is not Borel; for
surjective maps and invariant Borel sets, however, we can establish
measurability using Lusin's classic Separation Theorem. 
\BeginLemma{is-borel}
Let $X, Y$ be Polish, $f: X\to Y$ Borel measurable and onto, and assume $A\in\Inv{\Kern{f}}{\Borel{X}}$. Then
$\Bild{f}{A}$ is a Borel set in $Y$. 
\EndLemma

\BeginProof
See~\cite[Corollary 2.6]{EED-Alg-Prop_effFncts}.
\EndProof

These are two easy consequences. 

\BeginCorollary{is-borel-conseq}
Let $X, Y$ be Polish, $f: X\to Y$ measurable and onto, then 
$
\Inv{\Kern{f}}{\Borel{X}} = \{\InvBild{f}{B} \mid  B\in\Borel{Y}\}.
$
\EndCorollary

\BeginProof
Since $f$ is Borel measurable, and since $\InvBild{f}{B}$ is an
$\Kern{f}$-invariant Borel set of $X$, we obtain
$\Inv{\Kern{f}}{\Borel{X}}\supseteq \{\InvBild{f}{B} \mid
B\in\Borel{Y}\}$. On the other hand, if $A\subseteq X$ is
$f$-invariant, we have $\InvBild{f}{\Bild{f}{A}} = A$. Since
$\Bild{f}{A}$ is a Borel set by Lemma~\ref{is-borel}, the other
inclusion follows.
\EndProof

Given a map $f: X\to Y$, we will sometimes need to extend this to a map $f\times id_{[0, 1]}$, which sends $\langle x, q\rangle$ to $\langle f(x), q\rangle$, and we will have to know something about the kernel of this map, for which obviously 
\begin{equation*}
\isEquiv{\langle x, q\rangle}{\langle x', q'\rangle}{\Kern{f\times id_{[0, 1]}}}
\text{ iff }
f(x) = f(x') \text{ and } q = q'
\end{equation*}
holds. The $f\times id_{[0, 1]}$-invariant sets are described here.

\BeginLemma{is-borel-conseq-lem}
Let $X$ and $Y$ be Polish and $f: X\to Y$ be measurable and onto. Then 
\begin{equation*}
\Inv{\Kern{f\times id_{[0, 1]}}}{\Borel{X\otimes[0, 1]}} = \Inv{\Kern{f}}{\Borel{X}}\otimes \Borel{[0, 1]}.
\end{equation*}
\EndLemma

\BeginProof
See~\cite[Lemma 2.10]{EED-Alg-Prop_effFncts}.
\EndProof

We will now explore some structural properties which are induced by
surjective and measurable maps. It is somewhat surprising that this
induces an isomorphism on the set of all subprobabilities. Looking at the proof, it is even more
surprising that this is a consequence of Lusin's 
Theorem. 
The next lemma is a step towards the isomorphism.

\BeginLemma{gen-1}
Let $X$ and $Y$ be Polish, $f: X\to Y$ onto and measurable. Then the
$\sigma$-algebras $\Inv{\Kern{f}}{\Borel{X}}$ and $\Borel{Y}$ are isomorphic as
Boolean $\sigma$-algebras.
\EndLemma

\BeginProof
By Lemma~\ref{is-borel}, $\Bild{f}{D}\in \Borel{Y}$, whenever $D\in
\Inv{\Kern{f}}{\Borel{X}}$. 
Define
\begin{equation*}
\psi: 
  \begin{cases}
    \Borel{Y} & \to \Inv{\Kern{f}}{\Borel{X}}\\
D & \mapsto \InvBild{f}{D}.
  \end{cases}
\end{equation*}
Then $\psi(D)\in \Inv{\Kern{f}}{\Borel{X}}$ on account of $f$ being
measurable, and $\psi$ is injective because $f$ is onto. Now $C =
\InvBild{f}{\Bild{f}{C}}$ for $C\in \Inv{\Kern{f}}{\Borel{X}}$, and
$\Bild{f}{C}\in \Borel{Y}$, hence $\psi$ is onto as well. We know that
$\Bild{f}{\InvBild{f}{D}} = D$, since $f$ is onto, hence $f:
\Inv{\Kern{f}}{\Borel{X}}\to \Borel{Y}$ and $\psi: \Borel{Y} \to
\Inv{\Kern{f}}{\Borel{X}}$ are inverse to each other.
\EndProof

This will have some interesting consequences for the measurable
structure of the set of all subprobability measures. They are
introduced next.  

\medskip

We write $\SubProb{S, {\cal S}}$ for the set of all subprobability
measures on the measurable space $(S, {\cal S})$. This space
is made a measurable space upon taking as a $\sigma$-algebra
\begin{equation}
\label{intro-weak-sigma}
w({\cal S}) := \sigma\bigl(\{\basS{(S, {\mathcal S})}{A}{\bowtie q} \mid A \in
{\cal S}, q \in [0, 1]\}\bigr).
\end{equation}
Here 
\begin{equation*}
  \basS{(S, {\mathcal S})}{A}{\bowtie q} := 
\basS{\mathcal S}{A}{\bowtie q}
:= \{\mu \in \SubProb{S, {\cal S}} \mid \mu(A) \bowtie q\}
\end{equation*}
is the set of all subprobabilities on $(S, {\cal S})$ which evaluate on the measurable set $A$ as
$\bowtie q$, where $\bowtie$ is one of the relations $\leq, <,\geq,
>$. This $\sigma$-algebra is sometimes called the
\emph{weak-*-$\sigma$-algebra}. 

A morphism  $f: (S, {\cal S}) \to (T, {\cal T})$ in the category of
measurable spaces induces a map 
$
\SubProbSenza{f}: \SubProb{S, {\cal S}}\to\SubProb{T,{\cal T}}
$
upon setting 
\begin{equation*}
(\SubProbSenza{f})(\nu)(B) := \nu(\InvBild{f}{B})
\end{equation*}
for $B \in \Borel{T, {\cal T}}$; as usual, $\SubProbSenza{f}$ is
sometimes written as $\SubProb{f}$. Because 
\begin{equation*}
\InvBild{(\SubProbSenza{f})}{\basS{{\cal T}}{B}{\bowtie q}} =
\basS{\cal S}{\InvBild{f}{B}}{\bowtie q},
\end{equation*}
 this map is $w({\cal
  S})$-$w({\cal T})$-measurable as well. Thus $\SubProbSenza$ is an
endofunctor on the category of measurable spaces with measurable maps
as morphisms; in fact, it is the functorial part of a monad which is
sometimes called the \emph{Giry monad}~\cite{Giry}.

A ${\cal S}$-$w({\cal T})$-measurable map $K: S \to \SubProb{T,
  {\cal T}}$ is called a \emph{sub Markov kernel}, or sometimes a
\emph{stochastic relation} and denoted by $K: (S, {\cal S})\Trans (T,
{\cal T})$. From the definition it is apparent that a map $K: S\to
\SubProb{T, {\cal T}}$ is a stochastic relation iff these conditions
are satisfied:
\begin{enumerate}
\item the map $s\mapsto K(s)(D)$ is measurable for each $D\in{\cal
    T}$,
\item the map $D\mapsto K(s)(D)$ is a subprobability measure on ${\cal
    T}$ for each $s\in S$. 
\end{enumerate}

Returning to invariant sets, we note

\BeginLemma{gen-2}
Let $X$ and $Y$ be Polish, $f: X\to Y$ onto and measurable. Then there exists for
$\nu\in\SubProb{Y}$ a measure $\mu\in\SubProb{X, \Inv{\Kern{f}}{\Borel{X}}}$ such
that $\nu = \SubProb{f}(\mu)$. 
\EndLemma

\BeginProof
Put $\mu(A) := \nu(\Bild{f}{A})$ for $A\in \Inv{\Kern{f}}{\Borel{X}}$,
then $\mu\in\SubProb{S, \Inv{\Kern{f}}{\Borel{X}}}$ by
Lemma~\ref{gen-1}, and plainly $\SubProb{f}(\mu) = \nu$.
\EndProof

It follows from Corollary~\ref{is-borel-conseq} that a surjective map
$f:X \to Y$ between Polish spaces is final, and we obtain
from~\cite[Corollary 2.7]{EED-Alg-Prop_effFncts} that $\SubProb{f}:
\SubProb{X, \Inv{\Kern{f}}{\Borel{X}}}\to \SubProb{Y}$ is a Borel
isomorphism, when both spaces carry the weak $\sigma$-algebra.  

It may be interesting to compare the proof just given with the one
provided for the same fact in~\cite[Proposition 1.101]{EED-Book}. That
proof is based on an observation on the universally measurable right
inverse of a measurable map~\cite[Theorem 3.4.3]{Arveson}, which in
turn is based on a selection argument. The proof presented here makes
heavy use of the finality of surjective measurable maps which is
based essentially on Lusin's Separation Theorem.

\paragraph{Upper Closed Sets.}
Call a subset $V$ of the powerset of some set \emph{upper closed} iff
$A\in V$ and $A\subseteq B$ implies $B\in V$. Denote by
\begin{equation*}
  \Vau{S} := \{V\in w({\cal S}) \mid  V\text{ is upper closed}\}
\end{equation*}
all upper closed subsets of the weakly measurable sets of $\SubProb{S,
  {\cal S}}$. If $f: S\to T$ is ${\cal S}$-${\cal T}$-measurable,
define 
\begin{equation*}
  \Vau{f}(V) := \{W\in w({\cal T}) \mid \InvBild{(\SubProbSenza
      f)}{W}\in V\}.
\end{equation*}
Then $\Vau{f}: \Vau{S}\to \Vau{T}$. 

\paragraph{Some Conventions.}
\label{sec:some-conventions}
From now on, we will not write down explicitly the $\sigma$-algebra
${\cal S}$ of a measurable space $(S, {\cal S})$, unless there is good
reason to do so; if we need the $\sigma$-algebra of the measurable
space $S$, we refer to it as $\sigalg(S)$. Furthermore the space
$\SubProb{S}$ of all subprobabilities will be understood to carry the
weak-*-$\sigma$-algebra $w(\sigalg(S))$ always. 
When we consider $\weak(S)$ as a measurable space, its \sig-algebra
will always be the hit \sig-algebra. And whenever we consider a
measurable map $f:S\to T$ where $T$ has a hit \sig-algebra, we will
say that $f$ is \emph{hit-measurable}.
In an unambiguous context, e.g., whenever $S$ is a topological space with
$\sigalg(S)$ the Borel sets $\Borel{S}$ of $S$, we will write $w(S)$ rather than
$w(\sigalg(S))$. We will sometimes write
$\Sigma_{\rho}$ for $\Inv{\rho}{{\cal S}}$, and $\Sigma_{\Kern{f}}$
will be abbreviated as $\Sigma_{f}$.

\subsection{Some Indispensable Tools}
\label{sec:indispensable}

We post here for the reader's convenience some measure theoretic tools
which will be used all over; the reader may wish to
consult~\cite{EED-Meas} for more information and a tutorial on measures. Fix a set $S$. 

\paragraph{Dynkin's $\pi$-$\lambda$-Theorem.}
This technical tool is most useful when it comes to determine the $\sigma$-algebra generated by a family of sets~\cite[Theorem 10.1]{Kechris}. 

\BeginTheorem{pi-lambda}
Let $\mathcal{A}$ be a family of subsets of $S$ that is closed under finite intersections. Then $\sigma(\mathcal{A})$ is the smallest family of subsets containing $\mathcal{A}$ which is closed under complementation and countable disjoint unions. 
\QED
\EndTheorem

\paragraph{Choquet's Representation.}
The following condition on product measurability and an associated
integral representation attributed to Choquet is used~\cite[Corollary
3.4.3]{Bogachev}. Assume that $(S, {\mathcal S})$ is a measurable
space.
\BeginTheorem{Choquet}
Let $f: S \to \pReal$ be measurable and bounded, then
\begin{equation*}
C_{\bowtie}(f) := \{\langle s, r\rangle \in S \times \pReal \mid f(s) \bowtie r\} \in {\mathcal S}\otimes\Borel{\pReal}.
\end{equation*}
If $\mu$ is a $\sigma$-finite measure on ${\mathcal S}$, then
\begin{equation}
\label{Choquet-2}
\int_S f(s)\ \mu(dx)  = \int_0^\infty \mu(\{s \in S \mid f(x) > t\})\ dt
= (\mu\otimes\lambda)(C_{>}(f)).
\end{equation}
with $\mu\otimes\lambda$ as the product of $\mu$ with Lebesgue measure $\lambda$.
\QED
\EndTheorem
For $S$ an interval in $\Real$, the set 
$
C_>(f) = \{\langle s, t\rangle \in S \times \pReal \mid 0 \leq t < f(s)\}
$
may be visualized as the area between the $x$-axis and the graph of $f$. Hence formula~(\ref{Choquet-2}) specializes to the Riemann integral, if $f: \pReal\to\pReal$ is Riemann integrable, and $\mu$ is also Lebesgue measure.

\paragraph{Measurable Selections.}
Given a measurable space $S$ and a Polish space $X$, let $F(s)
\subseteq X$ be a non-empty closed subset of $X$ for all $s\in
S$. Call a sequence $\Folge{f}$ of measurable maps $f_n: S \to X$ a
\emph{Castaing representation for $F$} iff the set $\{f_n(s) \mid n
\in \Nat\}$ is dense in $F(s)$ for each $s\in S$. The Kuratowski and
Ryll-Nardzewski Selection Theorem~\cite[Corollary 6.9.4]{Bogachev}
states a condition on the existence of a Castaing representation for
$F$.

\BeginTheorem{Himmelberg}
Given a measurable space $S$ and a Polish space $X$, let $F(s)
\subseteq X$ be a non-empty closed subset of $X$ for all $s\in S$ such
that the set $ \{s\in S \mid F(s)\cap G \not= \emptyset\} \in
\sigalg(S)$ for all $G\subseteq X$ open. 
Then $F$ has a Castaing representation. \QED
\EndTheorem

\section{The Main Characters: Nondeterministic Kernels And Effectivity
  Functions}
\label{sec:main-char-nond-kern}
In this section we introduce the notion of \emph{stochastic
  nondeterminism} we are interested in, suitably defined so that
measurable concerns arising in a continuous setting are taken care
of. A model supporting this kind of behavior is that of
\emph{nondeterministic labelled Markov process}. Afterwards, we move a
step forward in complexity and consider \emph{stochastic effectivity
  functions}, which provide the building blocks for \emph{game models}
as developed in \cite{EED-GameLogic-TR}.

\subsection{Nondeterministic Labelled Markov Processes}
Nondeterministic labelled Markov processes (NLMP) were developed in
\cite{DWTC09:qest,Wolovick} as a nondeterministic variant of the
\emph{labelled Markov processes} studied by various authors
\cite{Desharnais,Desharnais-Edalat-Panangaden,EED-CongBisim}, a development
which generalized the Markov decision processes investigated by Larsen
and Skou~\cite{Larsen+Skou} for discrete state spaces to general
measurable spaces. 

Labelled Markov processes (LMP) are exactly the stochastic Kripke
\emph{frames}: that is, a state space $S$ with a labelled family of
stochastic relations or Markov kernels $K_a:S\Trans S$ working as
probabilistic transitions for each action $a$. NLMPs arise naturally,
e.g., by abstraction or underspecification of LMP. The nondeterminism
makes its appearance in internal form, in such a way that there is an
additional branching apart from that given by labels or actions. More
precisely, NLMPs allow, for each state $s$ and each action $a$, a
possibly infinite set  $\ka_a(s)$ of probabilistic behaviors. Models
combining probabilistic behavior with internal nondeterminism were
studied previously in the literature; \emph{probabilistic automata}
(as proposed by Segala
\cite{DBLP:conf/concur/Segala95,DBLP:journals/njc/SegalaL95,Stoelinga02anintroduction})
provide an example, although these specific models have only countable, hence
discrete, state spaces.

As in the case of LMP, the main ingredient of the definition of their
nondeterministic counterpart is the corresponding notion of
\emph{nondeterministic kernels} $\ka_a$; this requires a suitable
notion of measurability as well. The one we choose, \emph{hit-measurability}, specializes to the usual notion of measurability for the
deterministic case. A model for this approach comes from classical
topology, where the transition from a topological space to its hyper
space of closed subsets is supported by the Vietoris topology and its
variants~\cite{Kechris}.

In this paper, the focus will be on nondeterministic kernels, which are defined as follows.
\begin{definition}\label{def:nondet-kernel}
  Let $S$ be a measurable space. A \emph{nondeterministic
    kernel} on $S$ is a hit-measurable map $\ka : S \to \weak(S)$.
  We call $\ka$  \emph{image-finite} if the sets
  $\ka(s)$ are finite  for all $s\in S$,  \emph{image-countable}
  maps are defined analogously. 
\end{definition}

Consequently, $\ka(s)$ is for each $s\in S$ a measurable subset of
$\SubProb{S}$ such that $\{s\in S\mid \ka(s)\cap
G\not=\emptyset\}\in\sigalg(S)$ for every measurable subset $G$ of
$\SubProb{S}$. 

Having defined our kernels, we fix a set $L$ of labels and write down a formal definition of NLMP:
\begin{definition}\label{def:NLMP}
  A \emph{nondeterministic labelled Markov process} (NLMP) is
  a tuple $\nlmp S = (S, \{\ka_a : a\in L \})$ where the \emph{state
    space} $S$ is a measurable space, and for each \emph{label} $a\in
  L$, $\ka_a$ is a  nondeterministic
    kernel on $S$. We call $\nlmp S$
  \emph{image-finite (image-countable)} if all 
  of  its kernels $\ka_a$ are. 
\end{definition}

In \cite{DWTC09:qest,Terraf-Bisim-MSCS} the problem
of defining appropriate notions of bisimulation 
 for NLMP was addressed. State bisimulation for NLMP (called \emph{traditional bisimulation} in \cite{Terraf-Bisim-MSCS}) 
 is a generalization 
 of both the
 standard notion of bisimulation for nondeterministic processes,
 e.g.\ Kripke models, and of \emph{probabilistic bisimulation} of
 Larsen and Skou \cite{Larsen+Skou}, later studied in a more general
 measure theoretic context by Panangaden et al, see~
 \cite{Panangaden}. An alternative notion, \emph{event bisimulation}
 puts its emphasis on families of measurable subsets of the state
 space that are ``respected'' by the kernels; in the deterministic
 case, event bisimulations are included in the investigation of the
 lifting of countably generated equivalence relations to the space of
 probabilities \cite{EED-KleisliMorph}.

 Event bisimulations for a nondeterministic kernel are defined as
sub-$\sigma$-algebras on its state space. They yield a relation on the
state space by the construction \ref{eq:rel-from-sig-algebra}, and, by extension, a
relation on its subprobabilities.
\begin{definition}
Let $S$ be a measurable space. 
  \begin{enumerate}
  \item \label{def:nlmp_eb}
    An \emph{event bisimulation} on the nondeterministic kernel $\ka :
    S\to \weak(S)$ is a sub-$\sigma$-algebra $\Lda$ of $\sigalg(S)$ such that
    $\ka:(S,\Lda)\to(\weak(S),\hit(\weak(\Lda)))$ is measurable. We also say that a
    relation $R$ is an event bisimulation if there is an event
    bisimulation $\Lda$ such that $R=\mathcal{R}(\Lda)$.
  \item \label{def:nlmp_ls_sb}
    A relation $R$ is a \emph{state bisimulation}    on $\ka :    S\to \weak(S)$ if it is
    symmetric and for all $s,t\in S$, $s\mathrel{R}t$ implies that
    for all 
    $\mu\in \ka(s)$ there exists $\mu'\in \ka(t)$ such that
    $\mu\bR\mu'$.
  \end{enumerate}
  We say that $s,t\in S$ are state (resp., event-) \emph{bisimilar}, denoted
  by $s \sbisim t$ ($s \ebisim t)$, if there is a
  state (event) bisimulation $R$ such that $s\mathrel{R}t$.
\end{definition}

State bisimulations for NLMP are exactly the relations on the
state-space that satisfy the above definition simultaneously for each
kernel $\ka_{a}$ of the NLMP. This holds analogously for event
bisimulations. Since $\mathcal{R}(\Lda)$ is an equivalence, 
every event bisimulation yields an 
equivalence relation. This is not true for state bisimulations; in
case $R$ is a state bisimulation that is actually an equivalence
relation, we will say that $R$ is a \emph{state bisimulation
  equivalence}.

While state bisimulations will have a straightforward generalization
for effectivity functions, the r\^ole of event-based bisimulations will become apparent
only when subsystems are treated in Section~\ref{sec:coalgebraic-approach}.

\subsubsection{The Category of Nondeterministic Kernels}
In this section we will give a categorical presentation of
nondeterministic kernels. It is then immediate how to 
extend this to obtain a  category of NLMPs. Recall that the map
$S\mapsto\sigalg(S)$ can be regarded as a subfunctor of the
contravariant power set functor, acting on arrows by taking inverse
images. Formally, given  a measurable $f:S\to T$, we define $\sigalg
f:\sigalg(T)\to\sigalg(S)$ as
\begin{equation*}
\sigalg f (Q) := \InvBild{f}{Q}.
\end{equation*}
Then the map $S\mapsto\weak(S)$ is the composition $\sigalg\circ\subp$
of the Giry functor $\SubProbSenza$ and $\sigalg$. On arrows it behaves like this:
\begin{equation*}
\weak(f)(H) = \InvBild{(\SubProbSenza f)}{H}
\end{equation*}
for $H$ a measurable set of measures.
\begin{definition}\label{def:nlmp-morphism}
  The \emph{category $\NK$\MMP{Font changed, just testing} of nondeterministic kernels} has as objects all hit-measurable
  maps $\ka:S\to\weak(S)$ for some measurable space $S$. Given such  $\ka$
   and $\ka':T\to \weak(T)$, a $\NK$-morphism is a measurable
  map $f:S\to T$ such that the following diagram commutes:
  \begin{equation*}
    \xymatrix{
      S \ar[rr]^f \ar[d]_\ka  && T \ar[d]^{\ka'} \\
      \weak(S) 
      &&  \weak(T) \ar[ll]^{\weak(f)}}
  \end{equation*}
  that is, for all $s\in S$, 
  \begin{equation}\label{eq:diag-NLMP-morphism}
    \ka(s) = \InvBild{(\SubProbSenza f)}{\ka'(f(s))}.
  \end{equation}
\end{definition}
It is easy to see that this defines a category. 

Then the category of NLMPs (for a fixed set of labels $L$) has as objects all NLMP, and given $
\nlmp S = (S,
\{\ka_a : a\in L \})$ and $ \nlmp{S'} = (S', \{\ka'_a : a\in L \})$, a
morphism $f: \nlmp S \to \nlmp{S'}$ is a measurable map $f:S\to S'$
such that $f:\ka_a\to \ka'_a$ is a $\NK$-morphism for all $a$.

It should be noted that although a nondeterministic kernel on $S$ is a map from
$S$ to $\weak(S)$, we can't say that it is a coalgebra: In spite of $\weak(S)$ being a measurable space (with the hit
\sig-algebra), $\weak$ is not an endofunctor on the category of
measurable spaces. 

The first satisfactory test for this notion of morphism is that so
related nondeterministic kernels are state bisimilar (cf. \emph{zig-zag}
morphisms for LMP).

\BeginDefinition{def-direct-sum}
  Let $\ka:S\to \weak(S)$ and $\ka':T\to\weak(T)$ be nondeterministic kernels with
  $S$ and $T$ disjoint. The \emph{direct sum} $\ka\oplus \ka'$ is the
  nondeterministic kernel defined on the direct sum space $S\oplus T$ by the following
  stipulation
  \begin{equation*}
(\ka\oplus \ka')(x):=
                    \begin{cases} 
                            \ka(x) & x\in S \\
                            \ka'(x) & x \in T
                    \end{cases}
 \end{equation*}
\EndDefinition

The proof of the following proposition is included in the more general
Theorem~\ref{p:strong-morphisms-bisim} and thus we omit it;  
its statement formalizes the relation between morphisms and bisimulations.
\BeginProposition{pr:morphism-kernels-bisim}
 (Notation as in Definition~\ref{def:nlmp-morphism}.) 
Let $G := \{\langle s, f(s)\rangle\mid s\in S\}$ be the graph of $f$
with converse $G^{\smallsmile}$ Then $R := G\cup G^{\smallsmile}$  is a state bisimulation on $S\oplus T$.
\EndProposition
\begin{proof}
  It's easy to see that $R$-closed subsets $A$ of the disjoint union
  $S\oplus T$  are characterized by 
  \begin{equation}\label{eq:ker-f-closed}
     A = f^{-1}[A] \uplus f[A].
  \end{equation}
  To check that $R$ is a state bisimulation, we
  should see that  for any $s\mathrel{R}t$, and for every $\mu\in
  \ka(s)$ there exists a $\nu\in \ka'(t)$ such that $\mu\mathrel{R}\nu$.

  Assume $s\in S$, $s\mathrel{R}t$ and $\mu\in
  \ka(s)$. Then $t = f(s) \in T$  and we must find $\nu$ as above. We have a
  perfect candidate for that. By Equation~(\ref{eq:diag-NLMP-morphism}), $\subp
  f(\mu)\in \ka'(f(s)) = \ka'(t)$. Let's check $\mu\mathrel{\bR}\subp f (\mu)$. Take a
  $R$-closed set $A$. We should have 
  \begin{equation*}
    \mu(A) = \subp  f(\mu) (A) = \mu(f^{-1}[A]),
  \end{equation*}
  the last equality by definition of $\subp f$. Now $\mu(A) = \mu(
  f^{-1}[A])+  \mu(f[A])$ by Equation~(\ref{eq:ker-f-closed}). But since $\mu$ is
  supported on $S$, we have $\mu(f[A]) = 0$. Hence $\mu(A) = \mu(
  f^{-1}[A])$ and we have our result. The other way is symmetric.    
\end{proof}

\subsection{Stochastic Effectivity Functions}
\label{sec:effectivity-functions}
While nondeterministic kernels assign to
each state $s$ in a measurable state space $S$ a measurable set $\kappa(s)$ of subprobabilities in a
measurable way, effectivity functions use a different
approach. In order to explain it, we assume that we see player Angel
in state $s$, and we assume that Angel has a strategy to achieve a set
$A$ of probabilities over $S$ when playing game
$\gamma$. We deal with sets of probabilities as the possible
distribution of the new state after $\gamma$ because we want to model
the game semantics probabilistically. Hence the new distribution may
be an element of $A$. But if $B$ is a set of probabilities which
contains $A$, Angel has a strategy to achieve $B$ as well, so the
\emph{portfolio}, which we define as the family of all sets for which Angel has a strategy
in state $s$ playing $\gamma$, is upward closed. We want the sets of
measures to be measurable themselves, hence the portfolio is an upper
closed subset of $w(\sigalg(S))$, so that we obtain a map $S\to
\Vau{S}$. But this is not yet sufficient, a further requirement has to
be added. Before stating it, we mention that we will have good reasons to separate the domain from
the space underlying the range for an effectivity function. Thus we assume
that there is a second measurable space $T$ for which the effectivity
function determines possible sets of distributions as quantitative
assessments. It helps to visualize $S$ as the space of inputs, and $T$
as the space of outputs, respectively.

Let $H \in \bSS{w({\sigalg(T)})}$ be a measurable subset of
$\SubProb{T}\times[0, 1]$ indicating such a \emph{quantitative assessment
  of subprobabilities}. A typical example could be
\begin{equation*}
\{\langle \mu, q\rangle \mid \mu \in \basS{T}{A}{\geq q},\ 0 \leq q \leq 1\} 
= 
\{\langle \mu, q\rangle \mid \mu(A) \geq q, \ 0 \leq q \leq 1\}
\end{equation*}
for some $A \in \sigalg(T)$, combining all
subprobabilities and reals such that the probability for the given set
$A$ of states do not lie below this value. Consider a map
$P: S \to \Vau{T}$, fix some real $q$ and consider the
\emph{horizontal section} $H_q = \{\mu \mid \langle \mu, q\rangle \in
H\}$ of $H$ at $q$, viz., the set of all measures evaluated through
$q$. We ask for all states $s$ such that $H_{q}$ is \emph{effective
  for $s$}, i.e., $ \{s \in S \mid H_q \in P(s)\}$ to be  a measurable
subset of $S$.
 It turns out, however, that this is
not yet enough, we also require the real components being captured
through a measurable set as well --- after all, the real component
will be used to be averaged over later on, so it should behave
decently. This idea is captured in the following definition.

\BeginDefinition{t-measurability}
Call a map $P: S \to \Vau{T}$ \emph{t-measurable} iff 
\begin{equation*}
\{\langle s, q\rangle \mid H_q\in P(s)\} \in \bSS{\sigalg(S)}\}
\end{equation*}
whenever $ H \in\bSS{w({\sigalg(T)})}$. 
\EndDefinition

Thus if $P$ is t-measurable then we know in particular that all pairs of states and
numerical values indicating the effectivity of the evaluation of a
measurable set $A\in {\sigalg(T)}$, i.e., the set $ \{\langle s,
q\rangle \mid \basS{\sigalg(T)}{A}{\bowtie q} \in P(s)\} $ is always a
measurable subset of $S\otimes[0, 1]$. This is so because we know that
\begin{equation*}
\{\langle \mu, q \rangle \mid \mu\in \basS{\sigalg(T)}{A}{\bowtie q}, 0 \leq q
\leq 1\}
=
\{\langle \mu, q\rangle \in \SubProb{\sigalg(T)}\times[0, 1]\mid \mu(A) \bowtie q\},
\end{equation*}
and the latter set is a member of $\bSS{w({\sigalg(T)})}$ by Theorem~\ref{Choquet}.

For illustration, a description of t-measurability as Borel measurability in terms
of a compact Hausdorff topology on the Borel sets of $\SubProb{T}$ is
given. We assume for this the target space $T$ to be countably generated.
{
  \newcommand{\bars}[1]{\ensuremath{\|#1\|}} Let $X$ be a measurable
  space with a countably generated $\sigma$-algebra, say, $\sigalg{X}
  = \sigma(\{G\mid G\in \mathcal{G}\})$ with $\mathcal{G}$
  countable. We define on $\PowerSet{\weak(X)}$ the \emph{Priestley
    topology} $\tau$
which has as a sub basis 
\begin{align*}
 \mathfrak{B}:= \{&\bars{\basS{X}{A}{\leq q}}\mid A\in\mathcal{G}, q\in[0, 1]\text{
    rational}\}\ \cup\\
\{&-\bars{\basS{X}{A}{\leq q}}\mid A\in\mathcal{G}, q\in[0, 1]\text{
    rational}\}
\end{align*}
with
\begin{align*}
  \bars{W} & := \{{\cal T}\subseteq\PowerSet{w({\SubProb{X}}})\mid W\in {\cal T}\},\\
-\bars{W} & := \{{\cal T}\subseteq\PowerSet{w({\SubProb{X}}})\mid W\not\in {\cal T}\}.
\end{align*}
We obtain from Goldblatt's Theorem~\cite{Goldblatt_TopProofs}
\BeginTheorem{goldblatt-rasiowa}
$\PowerSet{\weak(\sigalg(X))}$ is a compact Hausdorff space with
the Priestley topology $\tau$  generated by $\mathfrak{B}$. $\Vau{X}$ is compact, hence Borel
measurable. \QED
\EndTheorem

Now let $P: S\to \Vau{T}$ be a map such that the $\sigma$-algebra on $T$ is countably generated, then we have

\BeginLemma{t-meas}
If $P$ is t-measurable, then $P$ is $\sigalg(S)$-$\Borel{\Vau{T}}$-measurable.
\EndLemma

\BeginProof
1.  Define $\mathfrak{B}$ as a sub-basis for the Priestly topology on
$\PowerSet{\weak(\SubProb{T})}$ with countable generator $\mathcal{G}$
for $\sigalg(T)$, as above.  We have to show that $\InvBild{P}{H}\in
\sigalg(S)$, provided $H\in\Borel{\Vau{T}}$. Because $\Vau{T}$ is
compact, hence closed, and because the Borel sets of $\tau$ are
countably generated (since the topology is), the Borel sets of
$\Vau{T}$ are countably generated with $\mathfrak{B}\cap\Vau{S}$ as
its generator. Thus it is sufficient to show that $\InvBild{P}{H}\in
\sigalg(S)$, provided $H\in \mathfrak{B}$. This is enough since the set $
\{H\in \Borel{\Vau{T}} \mid \InvBild{P}{H}\in \sigalg(S)\} $ is a
$\sigma$-algebra; if it contains $\mathfrak{B}$, it contains all its
countable unions, so in particular all open sets.

2.
Now take $A\in \mathcal{G}$ and $q\in[0, 1]\cap\Rational$, then
\begin{align*}
  \InvBild{P}{\bars{\basS{T}{A}{\geq q}}} 
& = \{s \in S \mid P(s)\in \bars{\basS{T}{A}{\geq q}}\}\\
& = \{s\in S \mid \basS{T}{A}{\geq q}\in P(s)\}\\
& = \{\langle s, r\rangle \mid  \basS{T}{A}{\geq r}\in P(s)\}_{q} \in \sigalg{(S)}.
\end{align*}
By the same argument is can be shown that 
\begin{equation*}
\InvBild{P}{-\bars{\basS{T}{A}{\geq q}}} = \{\langle s, r\rangle \mid
\basS{T}{A}{\geq r}\not\in P(s)\}_{q} \in \sigalg{(S)}.
\end{equation*}
\EndProof
}

Returning to the general discussion, we introduce stochastic effectivity functions now.

\BeginDefinition{effFnct} 
Given measurable spaces $S$ and $T$, a \emph{stochastic effectivity
  function} $P: S \eTrans T$ from $S$ to $T$ is a t-measurable map $P:
S \to\Vau{T}$.
\EndDefinition

These are easy examples for stochastic effectivity functions, see~\cite[Section 3]{EED-Alg-Prop_effFncts}.

\BeginExample{ex-stoch-effct-fncts}
Let $K_{n}: S\Trans T$ be a stochastic relation for each $n\in\Nat$. Then 
\begin{align*}
  s &\mapsto \{A\in w(\sigalg T)\mid K_{n}(s)\in A\text{ for all }n\in\Nat\}\\
  s &\mapsto \{A\in w(\sigalg T)\mid K_{n}(s)\in A\text{ for some }n\in\Nat\}
\end{align*}
define stochastic effectivity functions $S\eTrans T$. In particular,
\begin{equation*}
  s \mapsto \{A\in w(\sigalg T)\mid K(s)\in A\}
\end{equation*}
defines a stochastic effectivity function for a stochastic relation
$K: S\Trans S$. Conversely, if $P(s)$ is always a principal ultrafilter in
$w(\sigalg T)$ for each $s\in S$  and  effectivity function $P:
S\eTrans T$, then $P$ is derived from a stochastic relation. The
interplay between stochastic relations and effectivity functions
generated by them is investigated in greater detail in~\cite{EED-GameLogic-TR}

Let $S$ be finite, say, $S = \{1, \dots, n\}$, and assume we have a transition system $\to_{S}$ on $S$, hence
a relation $\to_{S}{}\subseteq S\times S$. Put $\mathrm{succ}(s) := \{s'\in S\mid
s\to_{S} s'\}$ as the set of a successor states for state $s$. Define for
$s\in S$ the set of weighted successors
\begin{equation*}
  \ell(s):= \bigl\{\textstyle{\sum_{s'\in
      \mathrm{succ}(s)}\alpha_{s'}\cdot e_{s'}}\mid 
\Rational\ni\alpha_{s'}\geq 0 \text{ for } s'\in \mathrm{succ}(s), \textstyle{\sum_{s'\in \mathrm{succ}(s)}\alpha_{s'}}\leq 1\bigr\}.
\end{equation*}
The upward closed set
\begin{equation*}
P(s)  := \{A\in\Borel{\SubProb{\{1, \dots, n\}}}\mid \ell(s)\subseteq A\}
\end{equation*}
yields a stochastic effectivity function $S\eTrans S$. 
\EndExample

The last example will be scrutinized below, the first ones indicate
that there is some easy relationship between countable sets of
effectivity functions. This is supported by the following observation,
the proof of which is straightforward. 

\begin{lemma}\label{l:countable-union-EF}
  Let $S$ and $T$ be measurable spaces, and $\mathfrak{P}  = \{P_\gamma \st
  \gamma\in\Gamma\}$ is a  countable family of t-measurable maps
  from $S$ to $\cones(T)$. Then the maps $\bigcup \mathfrak{P}$ and
  $\bigcap\mathfrak{P}$ defined
  as
  \begin{align*}
    \bigl(\textstyle\bigcup \mathfrak{P}\bigr)(s) & :=\textstyle\bigcup \{P_\gamma(s)\mid  \gamma\in\Gamma\},\\
    \bigl(\textstyle\bigcap \mathfrak{P}\bigr)(s) & :=\textstyle\bigcap \{P_\gamma(s)\mid  \gamma\in\Gamma\}
  \end{align*}
  for each $s\in S$, are t-measurable.
\QED
\end{lemma}

Morphisms for effectivity functions are defined as pairs of measurable maps which preserve the structure. To be specific

\BeginDefinition{morph-eff-fnct}
Let $P: S\eTrans T$ and $Q: M\eTrans N$ be stochastic effectivity functions, then a pair $(f, g)$ of measurable maps $f: S\to M$ and $g: T\to N$ is called a \emph{morphism $(f, g): P\to Q$} iff this diagram commutes
\begin{equation*}
\xymatrix{
S\ar[rr]^{f}\ar[d]_{P} && M\ar[d]^{Q}\\
\Vau{T}\ar[rr]_{\Vau{g}}&&\Vau{N}
}
\end{equation*}
\EndDefinition
By expanding, we obtain
\begin{equation*}
  G\in Q(f(s)) \Longleftrightarrow \InvBild{(\SubProbSenza{g})}{G}\in P(s)
\end{equation*}
for all $s\in S$ and all $G\in w(\sigalg(N))$. A pair $(\alpha,
\beta)$ of equivalence relations $\alpha$ on $S$ and $\beta$ on $T$ is
called a \emph{congruence for $P:S\eTrans T$ } iff we can find a
stochastic effectivity function 
$P_{\alpha, \beta}:
\Faktor{S}{\alpha}\to \Faktor{T}{\beta}$ 
such that the pair
  $(\eta_{\alpha}, \eta_{\beta})$ of factor maps is a morphism $P\to
  P_{\alpha, \beta}$. Thus we obtain 
  \begin{equation*}
    G\in P_{\alpha, \beta}(\Klasse{s}{\alpha}) \Longleftrightarrow \InvBild{(\SubProbSenza{\eta_{\beta}})}{G}\in P(s),
  \end{equation*}
whenever $G\in w(\sigalg(\Faktor{T}{\beta}))$. Thus if $\alpha$ cannot
distinguish the elements $s, s'\in S$, $\beta$ cannot distinguish the elements
in $P(s)$ from those in $P(s')$. It will be necessary to distinguish the domain from the space
  underlying the range for an effectivity function, as we will see
  when discussing subsystems in
  Section~\ref{sec:coalgebraic-approach}. 

If $S = T$ and $M = N$, however, we talk about a morphism
  $f: P\to Q$ resp.\ about a congruence $\alpha$ on $P$. Then we have, e.g., for the congruence $\alpha$ the equivalence
  \begin{equation*}
    G\in P_{\alpha}(\Klasse{s}{\alpha}) \Longleftrightarrow  \InvBild{(\SubProbSenza{\eta_{\alpha}})}{G}\in P(s).
  \end{equation*}
  Of course, one could still work in this situation with a morphism
  $(f, g): P\to Q$. But this would not make too much sense, since one
  wants to study in this case the effect the map $f: S\to M$ has on
  the effectivity function, so one wants to know how $P$ operates on
  $\Vau{f}$; a similar argument holds for congruences. The situation
  changes, however, when we have different measurable structures on
  the same carrier set.

  Stochastic effectivity functions with their morphisms constitute a
  category $\EF$\MMP{Font changed, just testing}. Morphisms and congruences for stochastic effectivity
  functions are studied in depth in~\cite{EED-Alg-Prop_effFncts}.

Just to provide an illustration, we establish
\BeginProposition{kern-is-morphism}
Let $\ul{X}$ and $\ul{M}$ be standard Borel spaces with stochastic
effectivity functions $P$ on $\ul{X}$ and $Q$ on $\ul{M}$, and assume
that $f: P\to Q$ is a surjective morphism. Then $\Kern{f}$ is a
congruence for $f$.
\EndProposition

\BeginProof
We show that $f\times id_{[0, 1]}: \ul{X_f}\times[0, 1]\to M\otimes
[0, 1]$ is final, then we can apply~\cite[Proposition
3.14]{EED-Alg-Prop_effFncts}. But finality of this compound map follows
from finality of $f$ together with the observation that
$\ul{M}\otimes[0, 1]$ is standard Borel:
 \begin{align*}
  \InvBild{(f\times id_{[0, 1]})}{\Borel{M\otimes[0, 1]}}
& = \InvBild{(f\times id_{[0, 1]})}{\Borel{M}\otimes\Borel{[0, 1]}},&&
\text{since $M$ is standard Borel}\\
& = \InvBild{f}{\Borel{M}}\otimes\Borel{[0, 1]}\\
& = \Xf\otimes\Borel{[0, 1]},&& \text{since $f$ is final}.
\end{align*}
This implies the assertion.
\EndProof

This indicates that there is a close interplay between congruences and
surjective morphisms, which is actually what one would expect when
studying both.

\paragraph{Duality.}
As an illustration for the manipulation of effectivity functions, we
introduce dual effectivity functions. When interpreting game logics
through effectivity functions, one assumes that the game is
determined, which means that either Angel or Demon has a winning
strategy. This assumption entails that if Angel has no strategy for
achieving a set $A$, Demon will have a strategy for achieving the
complement $A^{c}$ of $A$, so Angel is not effective for $A$ iff Demon
is effective for $A^{c}$. This can be modeled through the dual
$\partial P$ of the effectivity function $P$, so that $\partial P$
will encode the sets for which Demon is effective. Because it does not
cost much more, we define the dual for effectivity functions $S\eTrans
T$ rather than for $S\eTrans S$.

\BeginDefinition{def:duality}
  Let $P:S\eTrans T$ be given. The \emph{dual $\partial
  P$ of $P$} is defined by
  \begin{equation*}
\partial P(s):= \{D\in\weak(T) \st D^c\notin P(s)\}.
\end{equation*}
\EndDefinition

It is straightforward to see that $\partial$ is an
involutive automorphism of the category of $\EF$, when defined as the
identity on $\EF$-morphisms. Just for the record:
\begin{prop}
  $\partial$ is an endofunctor of the category $\EF$ such that
  $\partial^2 = \mathrm{Id}_\EF$.
\QED
\end{prop}



The dual of $P$ will be helpful when interpreting the box operator in
the two level logic discussed in Section~\ref{sec:two-level-logic}. 

\section{Bisimulations}
\label{sec:bisimulations}

State bisimulations on effectivity functions are defined as binary
relations on the base space. We follow the lines of the notion of
bisimulation for games as spelled in \cite{Pauly-Parikh}, but
extending it to continuous spaces. For motivation, we assume first
that we are in a discrete setting and consider
an effectivity function $Q$ from $S$ to the set $\{V\subseteq
\PowerSet{S}\mid V\text{ is upper closed}\}$. This gives a scenario
very similar to the one considered in~\cite{Pauly-Parikh}. We call a
symmetric relation $R\subseteq S\times S$ a state bisimulation for $Q$
iff we have for each $\langle s, t\rangle\in R$ the following: Given a
set $A \in Q(s)$, we can find a set $B\in Q(t)$ such that for each
element of $B$ there exists a related element of $A$, so that we can
find for each $t'\in B$ an element $s'\in A$ with $\langle s',
t'\rangle \in R$. This is the  generalization of Milner's
definition to upward closed sets. We are working with states on the
one hand, and with distributions over states on the other hand, so we
have to adapt the definition to our scenario. The straightforward way
out is to define the relation for distributions, giving us the
relation $\bar R$ on $\SubProb{S}$. This observation suggests the
following definition of state bisimulation.

\BeginDefinition{state-bisim-eff-fncts}
  Let $P:S \eTrans S$ be an effectivity function. A relation
  $R\subseteq S\times S$ is a \emph{state bisimulation on} $P$ if it
  is symmetric and for all $s,t\in S$ such that $s\mathrel{R}t$, we
  can find for each $A\in P(s)$ a set $B\in P(t)$ with this property:
  for each $\nu\in B$ there exists $\mu\in A$ such that $\mu
  \mathrel{\bar R} \nu$.
\EndDefinition

In the following subsections we will investigate how this new notion
interacts with morphisms and provide a logical characterization for
effectivity functions satisfying certain finiteness assumptions.

\subsection{Effectivity Function Morphisms Induce  State Bisimulations.}
In this section, we introduce strong morphisms and  show that the
graph of a strong morphism 
induces a bisimulation on the direct sum of two effectivity
functions. This result is analogous to
Proposition~\ref{pr:morphism-kernels-bisim}, and the argument is very
similar,  so that it can be considered as a generalization. The generalization will be 
helpful in Section~\ref{sec:representability}, where the relation
between effectivity functions and nondeterministic kernels will be
studied in detail. 
\begin{definition}
  Let $P:S\eTrans S$ and $Q:T\eTrans T$ be effectivity
  functions with disjoint state spaces $S$ and $T$. The \emph{direct sum} $P\oplus Q$ is the effectivity
  function defined on the direct sum space $S\oplus T$ by the following
  stipulation
  \begin{equation*}
(P\oplus Q)(x):=
                    \begin{cases} 
                            P(x) & x\in S \\
                            Q(x) & x \in T
                    \end{cases}
\end{equation*}
\end{definition}
In the case of effectivity functions, a
strict notion of morphism seems to be necessary to obtain the results
pertaining bisimulations. This is inspired by the transition from
homomorphisms to strong homomorphisms for Kripke models in modal
logics by saying ``that relational links are preserved
from the source model to the target \emph{and back
  again}''~\cite[p. 58]{Blackburn-Rijke-Venema}. We do not work with
relational links but rather with upward closed subsets, so this
strengthening happens like this:
\BeginDefinition{def:strong-morphism}
  Let $P:S\eTrans S$ and $Q:T\eTrans T$ be effectivity
  functions. A measurable map $f:S\to T$ is  a  \emph{strong morphism}
  from $P$ to $Q$ if it is 
  surjective 
  and the following holds   for all $s\in S$ and all $A\in\weak(S)$:
  \begin{equation}\label{eq:strong-morphism}
    A \in P(s) \iff \InvBild{(\SubProbSenza{f})}{B}  \subseteq A\text{ for some } B\in Q(f(s)).
  \end{equation}
\EndDefinition

Next we will prove two  features of strong morphisms in the
setting of standard Borel spaces.  The first one shows that strong morphisms are indeed
effectivity function morphisms, while the second 
expresses a form of
``backwards surjectivity'' of strong morphisms.
\BeginLemma{lem:strong}
  Let $P:S\eTrans S$ and $Q:T\eTrans T$ be effectivity
  functions on the standard Borel spaces $S$ and $T$, and  assume that
  $f:S\to T$ is a   strong morphism from $P$ to $Q$. Then
  \begin{enumerate}
  \item $f$ is an effectivity function morphism from $P$ to $Q$;
  \item The following holds for all $s\in S$ and $A\in\weak(S)$:
    \begin{equation}\label{eq:strong-onto}
      A \in P(s) \implies B
      \subseteq \Bild{(\SubProbSenza{f})}{A} \text{ for some } B\in Q(f(s)).
    \end{equation}
  \end{enumerate}
\EndLemma

\BeginProof
  For the first item, we have to check that 
  \begin{equation*}
    W\in Q(f(x)) \Longleftrightarrow \InvBild{(\SubProbSenza{f})}{W}\in P(x)
  \end{equation*}
  for all $s\in S$ and all $W\in w(\sigalg(T))$. Assume first that $W\in
  Q(f(s))$. In the definition of strong morphism, take $A:=(\subp
  f)^{-1}[W]$. We have, trivially, $(\subp f)^{-1}[W] \subseteq A$;
  hence by (\ref{eq:strong-morphism}) we obtain $A\in P(s)$, that is $
  (\subp f)^{-1}[W]\in P(s)$. For the other direction, suppose
  $(\subp f)^{-1}[W]\in P(s)$. By (\ref{eq:strong-morphism}) we know
  there exists some $B\in Q(f(s))$ such that $(\subp f)^{-1}[B]
  \subseteq(\subp f)^{-1}[W]$. By the proof of Lemma~\ref{gen-1}, we know
  that this is equivalent to $B\subseteq W$ since $f$ is
  surjective. But then $W\in Q(f(s))$ by upper-closedness.

  For the second item,  assume $A\in P(s)$. Since
  $f$ is strong, there must exist some $B\in 
  Q(f(s))$ such that $(\subp f)^{-1}[B]
  \subseteq A$. But then we might apply $\subp f$ on both sides of
  this inclusion and obtain
  \begin{equation}\label{eq:6}
    B  \subseteq \subp f[(\subp f)^{-1}[B]]   \subseteq \subp f[A],
  \end{equation}
  where the first  inclusion holds since $\subp f$ is onto by
  Lemma~\ref{gen-2}.
\EndProof

We require in both parts of the proof that we work in standard Borel
spaces, because we want to have that the image of a surjective Borel
map between standard Borel spaces under the subprobability functor is
onto again. 


\BeginTheorem{p:strong-morphisms-bisim}
  Let $S$ and $T$ be standard Borel spaces, let $P:S\eTrans \cones(S)$
  and $Q:T\eTrans\cones(T)$ be effectivity 
  functions, and  $f:P\to Q$ a strong
  morphism. The relation
  $R := G\cup G^\smallsmile$, where $G=\{\<s,f(s)\> : s \in S\}$ is the
  graph of $f$, is a state bisimulation on $P\oplus Q$.
\QED
\EndTheorem
\begin{proof}
  It's easy to see that $R$-closed subsets $A$ of the disjoint union
  $S\oplus T$  are characterized by 
  \begin{equation}\label{eq:4}
     A = f^{-1}[A] \uplus f[A].
  \end{equation}
  Note that pairs in $R$ have exactly one component in $S$ and the
  other in $T$. To check that $R$ is a state bisimulation, we
  first consider the case $s\mathrel{R}t$ with $s\in S$ and $t\in
  T$. We have to show
  \[\forall X\in (P\oplus Q)(s)\,\exists Y\in (P\oplus Q)(t):\forall \nu \in Y \,\exists
  \mu \in X. (\mu \mathrel{\bar R} \nu).\]
  In the case considered we have $t=f(s)$ and we might simply write
  \begin{equation*}\label{eq:5}
    \forall X\in P(s)\,\exists Y\in Q(f(s)):\forall \nu \in Y \,\exists
    \mu \in X. (\mu \mathrel{\bar R} \nu).
  \end{equation*}
  Let $X\in P(s)$. Since $f$ is strong, there must exist some $Y\in
  Q(f(s))$ such that $Y  \subseteq (\subp f)[X]$ by (\ref{eq:strong-onto}).
  Now take any $\nu \in Y$. By the previous inclusion, there exists some $\mu\in X$
  such that $\subp f(\mu) = \nu$. Let's check $\mu\bR\nu$.  Take a
  $R$-closed set $A$. We should have 
  \[\mu(A) = \nu(A)  = \subp  f(\mu) (A) = \mu(f^{-1}[A]),\]
  where the last equality holds by definition of $\subp f$. Now $\mu(A) = \mu(
  f^{-1}[A])+  \mu(f[A])$ by Equation~(\ref{eq:4}). But since $\mu$ is
  supported on $S$, we have $\mu(f[A]) = 0$. Hence $\mu(A) = \mu(
  f^{-1}[A])$ and we have proved this case. 

  Now assume $t\mathrel{R}s$. As before, $t=f(s)$ and we have to prove
  \begin{equation*}\label{eq:7}
    \forall Y\in Q(f(s)) \,\exists X\in P(s):\forall     \mu \in X\,\exists
    \nu \in Y . (\mu \mathrel{\bar R} \nu).
  \end{equation*}

  Let $Y\in Q(f(s))$. Since $f$ is a morphism, $(\subp f)^{-1}[Y] \in
  P(s)$ and we choose this preimage as our $X$. Now take $\mu\in
  X$. By definition of $X$, $\nu := \subp f (\mu)\in Y$, and we can
  perform the same calculation that solved the other case. 
\end{proof}
It should be emphasized that in this proof only the 
consequences of the strength of $f$ appearing in
Lemma~\ref{lem:strong} are used. Nevertheless, the full notion will
play a r\^ole while representing nondeterministic kernels as effectivity functions
in the next section.

It
is helpful at this point of the development to compare morphisms
with strong morphisms. Let $f: P\to Q$ be a morphism, and $F$ be a
strong morphism between $P$ and $Q$. The morphism $f$ satisfies the
condition 
\begin{equation*}
  G\in Q(f(s)) \text{ iff }\InvBild{\SubProb{f}}{G}\in P(s),
\end{equation*}
or, equivalently, 
\begin{equation*}
  Q(f(s)) = \{G\mid \InvBild{\SubProb{f}}{G}\in P(s)\},
\end{equation*}
so that $Q(f(s))$ is completely determined by $P(s)$ through $f$. The
strong morphism $F$, however, has to satisfy the condition
\begin{equation*}
  A \in P(s) \text{ iff }\exists B\in Q(F(s)):
  \InvBild{\SubProb{F}}{B}\subseteq A,
\end{equation*}
which means that $P(s)$ is the smallest upper closed subset of
$\sigalg(\weak(S))$ which contains the set
$\{\InvBild{\SubProb{F}}{B}\mid B\in Q(F(s))\}$ of inverse images of
elements in $Q(F(s))$ under $\SubProb{F}$. This is the manner in which the back
condition for Kripke models carries over.  

\subsection{Generation From Kernels}
\label{sec:representability}
In this section we will study the problem of generating an
effectivity function from a nondeterministic kernel. Essentially we ask for conditions under which a
hit-measurable map gives rise to a t-measurable map. 

Given a measurable set of measures on $S$, one can build a canonical
upper-closed set in $\cones(S)$: the principal filter generated by it
in the ordered set $(\weak(S), \subseteq)$. 
\begin{definition}
  Let $S$ be a measurable space. Given $W\in\weak(S)$, the
   \emph{filter generated by $W$} is the set
  \begin{equation*}
    \fil{W} := \{U \in \weak(S) \st W\subseteq U\}.
  \end{equation*}
  If $P:S\to \cones(S)$ and there exists $\ka:S\to\weak(S)$ such that
  $P(s) = \fil{\ka(s)}$ for all $s\in S$, we say that $P$ is
  \emph{filter-generated by} $\ka$. We will also say that $P$ is
  \emph{based on} $\ka$.

  Finally, call a measurable map $\ka:S\to\weak(S)$ a
  \emph{generating kernel} if $\fil{\ka}$ is
  t-measurable.
\end{definition}

We can state more formally the issues we are concerned with in this
section:
\begin{enumerate}
\item Showing that $\ka$ is a nondeterministic kernel if an effectivity function $P$ is filter-generated by $\ka$.
\item Investigating whether or not $\fil{\ka}$ is an effectivity
  function whenever $\ka$ is a nondeterministic kernel.
\end{enumerate}
We'll see that although we can easily answer the first question in the
affirmative, the second one is a little more tricky to deal with.

First, we'll check that filter-generated effectivity functions are
always based on a nondeterministic kernel.


\begin{prop}
  Let $\ka:S\to\weak(S)$.  If $\fil{\ka}$ is t-measurable, then
  $\ka$ is hit-measurable.
\end{prop}
\BeginProof
  Since the complements of hit-sets generate $\hit(S)$, it suffices to
  show that the set $\ka^{-1}[\weak(S)\setminus\hitt_A]$ 
  is measurable for arbitrary  $A\in\weak(S)$. 
  \begin{align*}
    \ka^{-1}[\weak(S)\setminus\hitt_A] &= \{s\st  \ka(s) \cap A = \emptyset \}
    \\
    &= \{s\st   \ka(s) \subseteq (\subp(S)\setminus A)  \} \\
    &= \{s\st \subp(S)\setminus A \in \fil{\ka(s)}  \} \\
    &= \bigl\{\<s,q\> \st \bigl((\subp(S)\setminus A)\times[0,1]\bigr)_q \in \fil{\ka(s)} \bigr\}_0.
  \end{align*}
  But this last set is measurable, because it is the cut of a
  measurable set (the latter, by t-measurability of $\fil{\ka}$).
\EndProof

We now state a criterion for a mapping $\ka:S\to\weak(S)$  to generate
an effectivity function. It actually provides a representation of
t-measurability as a parametrized form of hit-measurability, relating
these forms of stochastic nondeterminism. We will see in
Proposition~\ref{too-bad-uncountable}, however, that t-measurability
is much stricter a condition than being generated by a nondeterministic kernel in
this way.

\BeginProposition{pr:parametrized-hit}
  Let $\ka:S\to\weak(S)$. $\fil{\ka}$ is t-measurable if and
  only if the map $\bar \ka(s,q) 
  \doteq \ka(s)\times\{q\}$ (i.e., $\ka\times\mathrm{id_{[0,1]}}$) is
  $\sigalg(S)$-$\hit(\weak(S)\otimes\B([0,1]))$ measurable.
\EndProposition

\begin{proof}
  Suppose
  $H\in\weak(S)\otimes\B([0,1])$. Then
  \begin{equation*}
    \bar \ka(s,q) \subseteq H   \iff  \ka(s)\times\{q\}\subseteq H 
    \iff \ka(s) \subseteq H_q 
    \iff H_q \in\fil{\ka(s)},
  \end{equation*}
  hence $\{\<s,q\> : \bar \ka(s,q) \subseteq H \} = \{\<s,q\> : H_q \in
  \fil{\ka(s)}\}$, and the equivalence follows from Lemma~\ref{l:H-meas_subseteq}.
\end{proof}

When one restricts the state spaces to be standard Borel, there is a
tight connection between the category of generating kernels and the
effectivity functions generated as filters by them. The definition of
$\fil{\cdot}$ can be extended to encompass arrows by stipulating
$\fil{f}:= f$ for every $\NK$-morphism, rendering it a functor with
some special properties.  We will state this more formally in the next
theorem. 

For its proof we need a lemma as a preparation. Here we go:

\BeginLemma{lem:NK-iff-strong}
  Let $S$ and $T$ be standard Borel spaces, and $\ka:S\to\weak(S)$ and
  $\ka':T\to\weak(T)$ generating  kernels. 
  A surjective measurable map $f:S\to T$ is a  strong morphism between $\fil{\ka}$
  and $\fil{\ka'}$ if and only if $f$ is a $\NK$-morphism between $\ka$
  and $\ka'$. 
\EndLemma

\begin{proof}
  We will show that the two notions involved are the same by
  successively considering equivalent formulations. We start by writing
  down the definition of strong morphism between  $\fil{\ka}$
  and $\fil{\ka'}$: for all $A$ and $s$,
  \begin{equation*}
A \in \fil{\ka(s)} \text{ iff }\exists B\in \fil{\ka'(f(s))} : (\subp f)^{-1}[B]  \subseteq A.
\end{equation*}
  By applying the definition of  $\fil{\cdot}$, this is equivalent to
  \begin{equation*}
\ka(s)\subseteq A  \text{ iff } \exists B : \underline{\ka'(f(s)) \subseteq B} \et
  (\subp f)^{-1}[B]  \subseteq A.
\end{equation*}
  Since $\subp f$ is onto, we know by the proof of Lemma~\ref{gen-2} that the
  underlined condition is equivalent to $(\subp f)^{-1}[\ka'(f(s))]
  \subseteq (\subp f)^{-1}[B]$. Hence the last displayed formula
  is equivalent to the next one.
  \begin{equation*}
\ka(s)\subseteq A  \text{ iff } \exists B : (\subp f)^{-1}[\ka'(f(s))]
  \subseteq (\subp f)^{-1}[B] \et  (\subp f)^{-1}[B]  \subseteq A.
\end{equation*}
This in turn is equivalent to
  \begin{equation*}
\ka(s)\subseteq A  \text{ iff } (\subp f)^{-1}[\ka'(f(s))]\subseteq A,
\end{equation*}
because of transitivity of the subset relation (for $\ent$) and by taking
  $B:=\ka'(f(s))$ (for $\tne$). Finally, this is the same as
  \begin{equation*}
\ka(s) = (\subp f)^{-1}[\ka'(f(s))],
\end{equation*}
since $A$ is universally quantified. This last equation is the
  definition of $\NK$-morphism.
\end{proof}

We are ready now to enter into a discussion of the functorial
properties of $\fil{\cdot}$.  By Lemma~\ref{lem:NK-iff-strong}, we
conclude at once that $\fil{\cdot}$ is well defined on arrows and is
full, because there is a one-one correspondence between $\NK$-arrows
and strong morphisms. Finally, $\fil{\cdot}$ trivially preserves
composition of morphisms since it is the identity on arrows. This establishes

\BeginTheorem{4-dot-eleven}
  $\fil{\cdot}$ is a full functor between the category of
  generating  kernels over standard Borel spaces and the one of
  filter-generated effectivity functions over standard Borel spaces
  with strong morphisms.
\QED
\EndTheorem

It can actually be proved that   $\fil{\cdot}$ is an isomorphism between the
categories above.

In the next proposition we state that $\fil{\cdot}$ preserves and
reflects state bisimulations. In order to make a distinction, we will
use the shorter name \emph{$\NK$-bisimulations} for the state
bisimulations on nondeterministic kernels. The proof is fairly
straightforward and requires only the definition of a bisimulation. 
\BeginProposition{p:fil-pres-bisim}
  Let $\ka:S\to\weak(S)$ be a generating  kernel and $R$ a
  binary relation on $S$. $R$ is a $\NK$-bisimulation on $\ka$ if and
  only if it is a bisimulation on $\fil{\ka}$.
\QED
\EndProposition



A very simple observation in the game theoretic context is that when
generating effectivity function $\fil{\ka}$ from $\ka$, we end up with
an effectivity function where there are essentially no choices for
Angel: we may say that $\fil{\cdot}$ \emph{demonizes} the set of
measures $\ka(s)$ for each $s$, since Demon is effective for each and
every element of $\ka(s)$. An alternative representation would be
given by ascribing each measure as effective for Angel; that is, each
singleton $\{\mu\}$ with $\mu\in \ka(s)$ should be effective for
Angel. Then this alternative representation would be given by the
\emph{angelization} functor
\[\bigl(\mathfrak{A}\ka\bigr)(s) :=  \textstyle\bigcup\{\fil{\{\mu\}} \st \mu\in \ka(s)\}.\]
This functor is defined as the identity on arrows. One may ask whether a
nondeterministic kernel $\ka$ is generating by $\fil{\cdot}$ if and only if
$\mathop{\mathfrak{A}}\ka$ is t-measurable. The answer follows easily
by using duals. We omit the proof of the following straightforward proposition.
\begin{prop}
  The functors $\fil{\cdot}$ and $\mathop{\mathfrak{A}}$ satisfy $\partial
  \mathop{\mathfrak{A}} = \fil{\cdot}$ and $\partial \fil{\cdot} =
  \mathop{\mathfrak{A}}$.
\QED
\end{prop}
A property expressed in terms of $\mathop{\mathfrak{A}}$ may be
rephrased as one about $\fil{\cdot}$, and vice versa. This also holds
for state bisimulations, though since it is not yet apparent how to
express them in terms of arrows, a proof is needed. %

The following fairly surprising observation indicates that t-measurability is much stricter
that hit-measurability: only image-countable kernels can be represented
as effectivity functions.
\BeginProposition{too-bad-uncountable}
  Let $S$ be a Polish space, $\ka:S\to \weak(S)$, and assume there
  exists $s_0\in S$  such that $\ka(s_0)$ is uncountable. Then
  $\fil{\ka}$ is not t-measurable.
\EndProposition
\begin{proof}{\def\xi{A}
  Let $\xi :=\ka(s_0)$. Since $S$ is Polish, $\subp(S)$ is standard
  Borel. $\xi\subseteq\subp(S)$ is an uncountable Borel set, and hence
  the trace 
  $
  (\xi,\weak(S)\cap \xi)
  $ 
  is an uncountable standard Borel
  space. Then we know there exists a measurable
  $H\subseteq \xi\times [0,1]$ such that $\pi_2[H]$ is an
  analytic non-Borel subset of $[0,1]$ (i.e.,
  $\pi_2[H]\notin\B([0,1])$), see~\cite[Theorem 4.1.5]{Srivastava}.

  If $\fil{\ka}$ were t-measurable, then $\{\langle s,q\rangle:
  (S\times[0,1]\setminus H)_q \in \fil{\ka(s)}\}$
  should belong  to $\B(S)\otimes\B([0,1])$, since
  $H^\comp = S\times[0,1]\setminus H\in\weak(S)\otimes\B([0,1])$. 
  Therefore, the set 
  appearing in the following calculation should be measurable:
  \begin{align*}
    (\{s_0\}\times [0,1])\setminus \bigl\{\<s,q\>\st  (H^\comp)_q \in
    \fil{\ka(s)}\bigr\} 
    & = \bigl\{\<s_0,q\>\st  (H^\comp)_q \notin \fil{\ka(s_0)}\bigr\} \\
    & = \bigl\{\<s_0,q\>\st  \ka(s_0) \nsubseteq (H^\comp)_q \bigr\} \\
    & = \bigl\{\<s_0,q\>\st  \xi \nsubseteq (H_q)^\comp \bigr\} \\
    & = \bigl\{\<s_0,q\>\st  \xi \cap H_q \neq \emptyset \bigr\} \\
    & = \bigl\{\<s_0,q\>\st  \exists \mu \in \xi :  \<\mu,q\> \in H\bigr\} \\
    & \stackrel{\text{(*)}}{=} \{s_0\} \times \pi_2[H].    
  \end{align*}
  where the equality (*) holds since $\pi_1[H]\subseteq \xi$. Since
  the last set has a non-Borel cut, it does not belong to 
  $\B(S)\otimes\B([0,1]) = \Borel{S\times[0, 1]}$, so $\fil{\ka}$ is not t-measurable.
}\end{proof}


We will show now that every image-finite kernel is generating. For
this, will need the following lemma, in which we check that every
image-finite kernel is given by stochastic relations. This is actually
an easy application of measurable selections.

\BeginLemma{l:finite-NLMP-Markov-kernels}
  Let $S$ be a Polish space and $\ka:S\to \weak(S)$ be a
  nondeterministic kernel
  such that for all $s\in S$, 
  $\ka(s)$ is finite. Then there exists a sequence $\Folge{K}$ of stochastic
  relations $K_n : S\Trans S$  such that $\ka(s) = \{K_n(s) \st n\in\N\}$
  for all $s$. 
\EndLemma

\begin{proof}
  The sets  $\ka(s)$ are closed  for all $s\in S$. 
Since $\ka$ is hit-measurable, we infer that for each open subset $G$
of $\SubProb{S}$ the set $\{s\in S \mid \ka(s)\cap G\not= \emptyset\}$
is measurable, thus there exists a Castaing representation $\Folge{K}$
for $\ka$ by the Selection Theorem~\ref{Himmelberg}. Clearly, each $K_n$ is a stochastic relation $S\Trans S$. 
\end{proof}
\begin{theorem}\label{th:image-finite-generating}
  Let $S$ be a Polish. Every image-finite
  nondeterministic kernel  $\ka:S\to \weak(S)$  is generating. 
\end{theorem}
\begin{proof}
  By Lemma~\ref{l:finite-NLMP-Markov-kernels}, there are stochastic
  relations $K_n : S\Trans S$ ($n\in\N$)  such that
  \[\fil{\ka(s)}  = \textstyle\bigcap\{\fil{\{K_n(s)\}}\st n\in\N\}.\]
  for all $s\in S$.
  Since each $\fil{\{K_n\}}$ is t-measurable, by
  Lemma~\ref{l:countable-union-EF} $\fil{\ka}$ is t-measurable.
\end{proof}

\subsection{Two-Level Logic} 
\label{sec:two-level-logic}

We now introduce a modal logic, $\two$, that characterizes state
bisimilarity for finitary effectivity functions. As the logic for NLMP
appearing in \cite{Terraf-Bisim-MSCS}, it is divided in two levels:
one devoted to states and another one to measures.  Two operators serve as
a bridge between the levels. 

In terms of a game theoretic scenario,  we want to express in this
way that in a particular state, Angel has a strategy to attain a set
of measures. This is done by using a modality $\dia$. On the other
hand, we have quantitative operators of the form $[\phi\bowtie q]$,
that allow us to single out measures, which assign a value $\bowtie q$
to the set of states defined by the state formula $\phi$.

The $\dia$-modality is patterned after a necessity operator in
traditional modal logics (compare~\cite{hansen2003monotonic,Venema-Handbook}), but
instead of asserting that Angel can  ensure a set of states, it does
so with respect a set of measures. There is an intimate relation between
$\two$ and the stochastic version of game logic developed in
\cite{EED-GameLogic-TR}, but we won't focus on that matter in the
present paper.

In the following, $\phi$ will range over state formulas and $\psi$
will range over measure formulas. The logic is given by the following grammar,
\begin{align}
  \phi & \ ::== \ \top \ \mid \ \phi_1\land\phi_2 \ \mid \ \dia\psi \ \mid \ \Box\psi \label{eq:10}\\
    \psi & \ ::== \ \psi_1\land\psi_2  \ \mid \ \psi_1\lor\psi_2  \ \mid
  \ [\phi\bowtie q], \nonumber
\end{align}
where $\bowtie$ is $<$ or $>$ and $0\leq q< 1$ is rational. 

Fix a standard Borel space $S$ and let $P: S\eTrans S$ be an
effectivity function. The
interpretation of the propositional operators over $S$ is straightforward, 
so we focus on the other connectives. We define inductively
\begin{align*}
  \sem{\phi} & := \{s\in S \mid s\models \phi\},\\
\sem{\psi} & := \{\mu\in\SubProb{S}\mid \mu\models \psi\}
\end{align*}
for state formulas $\phi$ and measure formulas $\psi$, then
\begin{align*}
s\models \dia\psi &\iff \sem{\psi} \in P(s)\\
s\models\Box\psi& \iff \sem{\psi}\in\partial P(s)\\
\mu \models [\phi\bowtie q] &\iff \mu(\sem{\phi})\bowtie q.
\end{align*}
As usual, $\Box$ is defined dually from $\dia$, using duality on $P$ here. We will call
\emph{$\two$-formulas} the ones obtained from the first line of
productions (\ref{eq:10}), i.e., the state-formulas. It is plain that
there are only countably many $\two$-formulas.

Our next step will be to prove that for every effectivity function
$P$, the relation of logical equivalence is smooth, in the sense made
precise by the following

\BeginDefinition{smooth-equiv-relation}
An equivalence relation $\rho$ on the standard Borel space $S$ is called
\emph{smooth} (or \emph{countably generated}) iff there exists a
countable set $\Folge{C}$ of Borel sets such that
\begin{equation*}
  \isEquiv{s}{t}{\rho} \quad\text{ if and only if }\quad  \forall n\in \Nat: s\in C_{n}\Longleftrightarrow t\in C_{n}.
\end{equation*}
\EndDefinition

These relations have some desirable properties, among others it is
known that the factor space becomes an analytic space, and, more
important for us, that 
\begin{equation}\label{eq:3}
\Inv{\rho}{\Borel{S}} = \sigma(\{C_{n}\mid
n\in\Nat\}).
\end{equation}
This is established through the Blackwell-Mackey Theorem,
see~\cite[Theorem 4.5.7]{Srivastava}, and~\cite[Proposition 2.108]{EED-Meas} for a
discussion; note again that this relies on $S$ being a standard Borel space,
because the Blackwell-Mackey Theorem makes use of Suslin's Theorem,
which is not available in general measurable spaces. 
The characterization of invariant sets shows that a smooth relation permits a concise description of
its invariant Borel sets.

This is a necessary step towards the proof of
completeness for bisimilarity.
\begin{lemma}
  For every $\two$-formula $\phi$, $\sem{\phi}$ is measurable.
\end{lemma}
\begin{proof}
  We will see this by structural induction, proving that formulas of
  both productions are measurable in the respective spaces ($S$ and
  $\subp(S)$).

  Again, the only nontrivial cases are the ones with new
  operators. Assume $\sem{\psi}$ is measurable, i.e.,
  $\sem{\psi}\in\weak(S)$.
  \begin{equation*}
\sem{\dia\psi} = \{s \st s\models \dia\psi\} =  \{s \st
  \sem{\psi}\in P(s)\} = 
   \{\<s,q\> \st (\sem{\psi}\times[0,1])_q \in P(s)\}_0,
\end{equation*}
  but this last set is measurable since $P$ is t-measurable. 

  Now assume inductively that $\sem{\phi}$ is measurable. We have
  \begin{equation*} 
\sem{[\phi\bowtie q]} = \{\mu \st \mu(\sem{\phi})\bowtie q\} =
  \basS{\sigalg{S}}{\phi}{\bowtie q},
\end{equation*}
  again a measurable set (now in $\subp(S)$).
\end{proof}

\BeginCorollary{cor:two-smooth}
  For every  effectivity function $P:S\eTrans S$, the relation of
  $\two$-equivalence,
\begin{equation*}
s \eqlog t \iff \forall \phi\in\two.(s\models \phi \sii t\models
  \phi),
\end{equation*}
  is smooth.
\EndCorollary
\begin{proof}
  Immediate by the previous lemma and the observation that there are
  only countably many $\two$-formulas.
\end{proof}

In the next lemma we show that state bisimulations preserve $\two$-formulas.
\begin{lemma}\label{l:sb-preserve-2L}
  Let $P:S\eTrans S$ be an effectivity function, and let $R$ be a
  state bisimulation for $P$. For every $\two$-formula $\phi$,
  $\sem{\phi}$ is $R$-closed.
\end{lemma}
\begin{proof}
0.
  The proof is done by showing inductively that
  \begin{enumerate}
  \item $\sem{\phi}$ is $R$-closed,
  \item $\sem{\psi}$ is $\bR$-closed.
  \end{enumerate}
  Since the family of $R$-closed sets are preserved by arbitrary
  Boolean operations, we only need to consider the modality $\dia$ and the
  $[\phi\bowtie q]$ construction.

1.
  Assume inductively that $\sem{\psi}$ is $\bR$-closed. Let's show
  $\sem{\dia\psi}$ is $R$-closed. Suppose $s\in\sem{\dia\psi}$, i.e.,
  $\sem{\psi}\in P(s)$, and $s\mathrel{R} t$. Since $R$ is a
  bisimulation, there exists $Y\in P(t)$ such that for any given
  $\nu\in Y$ there exists $\mu\in \sem{\psi}$ with $\mu \bR\nu$. 
  Since $\sem{\psi}$ is $\bR$-closed, we have $\nu\in\sem{\psi}$ for
  all $\nu\in Y$ and then $\sem{\psi}\in P(t)$. But this means
  $t\in\sem{\dia\psi}$, and we have this case.

2.
  Now assume inductively that $\sem{\phi}$ is $R$-closed; we'll see
  that $\sem{[\phi\bowtie q]}$ is $\bR$-closed. For this, assume
  $\mu\models[\phi\bowtie q]$ and $\mu\bR\nu$. But then
  $\mu(\sem{\phi})=\nu(\sem{\phi})$ and hence,
  \[\mu(\sem{\phi})\bowtie q \iff \nu(\sem{\phi})\bowtie q.\]
  This implies $\nu\models[\phi\bowtie q]$.
\end{proof}

We obtain immediately through Theorem~\ref{p:strong-morphisms-bisim} and
  Lemma~\ref{l:sb-preserve-2L}.
\begin{corollary}
  Strong morphisms preserve $\two$-formulas.
\QED
\end{corollary}

We are now in a position to establish a partial converse of the previous lemma.    $\two$ is
actually complete
for state bisimilarity on effectivity functions, which are
finitary in the sense made precise below. This concept generalizes the
notion of \emph{uniformly finitary effectivity functions}, already used
in the non-probabilistic context \cite{Pauly-Parikh}.
\begin{definition}\label{def:finitary}
  An effectivity function $P:S\eTrans S$ is   
  \begin{enumerate}
  \item \emph{finitary}
    if for
    each $s\in S$ there exists a finite family $E(s) := \{\ka_i(s) : i=1,\dots,n_s\}$
    with each $\ka_i(s)\subseteq\weak(S)$ finite such that 
    \begin{equation*}
P(s) = \textstyle\bigcup\{\fil{\ka_i(s)} : i =1,\dots,n_s\}.
\end{equation*}

  \item \emph{finitely supported} if it is finitary with $n_s=1$ for all
    $s\in S$.
  \end{enumerate}
\end{definition}
It is immediate from the definition that finitely supported
effectivity functions are exactly the ones which are filter-generated
by an image-finite kernel. Also, it may be interesting to note that both
finitary and filter-generated are related to the notion of
\emph{core-completeness} appearing in \cite{hansen2003monotonic}.

The next theorem generalizes one of the main results of
\cite{DWTC09:qest,Terraf-Bisim-MSCS}, the logical characterization of
bisimilarity for image-finite nondeterministic kernels, since every
such kernel encodes an effectivity function by
Theorem~\ref{th:image-finite-generating}. In the following, 
we use $\sem{\two}$ to denote the set of all extensions of
$\two$-formulas, that means, 
\[\sem{\two} = \{\sem{\phi} \st \phi \in \two\}.\]
\begin{theorem}
  Let $P: S\eTrans S$ be a finitary effectivity function. Two states that satisfy
  exactly the same $\two$-formulas are bisimilar.
\end{theorem}
\begin{proof}
  Since by Corollary~\ref{cor:two-smooth} the equivalence relation $\approx$ is smooth, we obtain
  $
    \Inv{\approx}{\sigalg(S)}=\sig(\sem{\two})
  $
from ~(\ref{eq:3}).  The
  result will follow if we show that $\approx$ is a state 
  bisimulation.

  We proceed by way of contradiction. That is, assume $s\approx t$ and
  there exists  $A \in P(s)$ so that for any given $B\in P(t)$ there
  exists $\nu\in B$ such that $\mu\centernot{\overline{\approx}} \nu$
  holds for all $\mu\in A$. 
  Since $P$ is finitary, there must exist a finite $A_0 = \{\mu_i:
  i\in I\}\in P(s)$ such
  that $A_0\subseteq A$. Also, there are finitely many (finite)
  $B_j$ such that $P(t) =  \textstyle\bigcup_j\fil{B_j}$. We enumerate
  the $\nu$'s accordingly. Hence we can find $A_{0}\in P(s)$ such that
  for all $B_{j}\in P(t)$ the following holds
  \begin{equation}
\exists \nu_j\in B_j \, \,\forall \mu_i\in
    A_0. (\mu_i\centernot{\overline{\approx}} \nu_j).
  \end{equation}
  Now, $\mu_i\centernot{\overline{\approx}} \nu_j$ if and only if
  there exists some $Q\in\Inv{\approx}{\sigalg(S)}$  such that
  $\mu_i(Q)\neq\nu_j(Q)$.  By (\ref{eq:3}), we may choose
  $Q\in\sig(\sem{\two})$ and hence $\mu_i$ and $\nu_j$ differ on
  $\sig(\sem{\two})$. But since $\sem{\two}$ is a generator of
  $\sig(\sem{\two})$ which is closed under finite intersections,  the $\pi$-$\lda$ Theorem~\ref{pi-lambda} ensures that  there
  are $\two$-formulas $\phi_{ij}$ witnessing the fact that $\mu_i$ and
  $\nu_j$ are different:
  \[ \nu_j(\sem{\phi_{ij}})\neq    \mu_i(\sem{\phi_{ij}}).\]
  Without loss of generality we might state that we can find $A_{0}\in
  P(s)$ such that for all $B_{j}\in P(t)$ this is true:
  \begin{equation*}
    \exists \nu_j\in B_j \, \,\forall \mu_i\in
    A_0. 
     \nu_j(\sem{\phi_{ij}}) \bowtie_{ij} q_{ij} \bowtie_{ij}
    \mu_i(\sem{\phi_{ij}})
  \end{equation*}
for some $q_{i, j}$. 
  This can be expressed by a $\two$-formula, as follows:
  \begin{equation*}
t\models\Box\textstyle\bigvee_j\textstyle\bigwedge_i [\phi_{ij} \bowtie_{ij} q_{ij}]
\end{equation*}
  But then $s$ does not satisfy this formula, and we reach a 
  contradiction since we assumed $s\approx t$. 
\end{proof}

The\MMP{new} proof shows a strong resemblance to the familiar proof of the
Hennessy-Milner Theorem, see,
e.g.~\cite[p. 69]{Blackburn-Rijke-Venema}, by exploiting the assumption
of being finitely supported, identifying for each of these finite
cases a finite number of culprits and then, using the logic's
finitary constructors, constructing one violating formula.

Since $\two$-equivalence is smooth, we obtain the following
\begin{corollary}\label{c:bisimilarity-smooth}
  For every finitary effectivity function $P$, state bisimilarity on $P$ is
  smooth.
\end{corollary}

In the next section, we will make a first approach into the study of
the notions of bisimilarity and behavioral equivalence from a
coalgebraic perspective. The main tool would be that of
\emph{subsystems}, a refined version of event bisimulations suited for
effectivity functions. The link to the material in the present section is
given by the last Corollary, since smooth bisimulations are proved to
induce subsystems.

\section{Subsystems: A Coalgebraic Approach}
\label{sec:coalgebraic-approach}
Given a measurable space $X$, assume that ${\cal C}$ is a
sub-$\sigma$-algebra of $\sigalg(X)$; denote the measurable space $(X,
{\cal C})$ by $X_{{\cal C}}$, so that $\sigalg(X_{{\cal C}}) = {\cal
  C}$. Since $\sigalg(X_{{\cal C}})\subseteq\sigalg(X)$, we see that
the identity $i_{{\cal C}}: X \to X_{{\cal C}}$ is measurable. Now
assume that $P: X \eTrans X$ is an effectivity function. The
restriction of $P$ to ${\cal C}$ focuses $P$ on the events described
in ${\cal C}$, provided this is algebraically possible. This leads to
the notion of a subsystem, formally:

\BeginDefinition{subsystem}
Given a stochastic effectivity function $P$ on the measurable space
$X$, a sub-$\sigma$-algebra $\mathcal{C}\subseteq\sigalg(X)$ defines a
\emph{subsystem of $P$} iff we can find a stochastic effectivity
function $P_{\mathcal{C}}$ on $X_{{\cal C}}$ such that
$i_{\mathcal{C}}$ defines a morphism $P\to P_{\mathcal{C}}$.
\EndDefinition

Assume that $\mathcal{C}$ defines a subsystem of $P$, then this
diagram commutes
\begin{equation}
\label{lab:subsyst}
\xymatrix{
X\ar[rr]^{i_{\mathcal{C}}}\ar[d]_{P}&&X_{{\cal C}}\ar[d]^{P_{\mathcal{C}}}\\
\Vau{X}\ar[rr]_{\Vau{i_{\mathcal{C}}}}&&\Vau{X_{\mathcal{C}}}
}
\end{equation}
This means that we have
\begin{equation*}
  P_{\mathcal{C}}(x) = \{D\in w(\sigalg{X}) \mid
  \InvBild{(\SubProbSenza i_{\mathcal{C}})}{D}\in P(x)\}
\end{equation*}
for each $x\in X$. Hence this can be used as a definition for $P_{\mathcal{C}}$. Note that
 $(\SubProbSenza i_{\mathcal{C}})(\mu)$ is the restriction of subprobability
$\mu\in\SubProb{X}$ to the $\sigma$-algebra
$\mathcal{C}$. The true catch is of course that $P_{\mathcal{C}}$ has
to be t-measurable.

So one might ask whether there exist subsystems at all. Before
answering this\MMP{This is new.}, we shall make a brief excursion into congruences for
stochastic relations. Let $X$ and $Y$ be standard Borel spaces and $K:
X \Trans Y$ be a stochastic relation. A congruence $(\alpha, \beta)$
for $K$ is a pair of smooth equivalence relations $\alpha$ on $X$ and
$\beta$ on $Y$ such that there exists a relation $K_{\alpha, \beta}:
\Faktor{X}{\alpha}\Trans \Faktor{Y}{\beta}$ such that this diagram
commutes:
\begin{equation*}
\xymatrix{
X\ar[d]_{K}\ar[rr]^{\fMap{\alpha}} && \Faktor{X}{\alpha}\ar[d]^{K_{\alpha, \beta}}\\
\SubProb{Y}\ar[rr]_{\SubProb{\fMap{\beta}}} && \SubProb{\Faktor{Y}{\beta}}
}
\end{equation*}
It is not difficult to see that this condition is equivalent to saying
that $K': (X, \Inv{\alpha}{X})\Trans (Y, \Inv{\beta}{Y})$ is a
stochastic relation, where $K'(x)$ is the restriction of $K(x)$ to the
$\sigma$-algebra $\Inv{\beta}{Y}$ of $\beta$-invariant
sets~\cite[Exercise~21]{EED-Meas}. Stating this formally, $K'(x) =
(\SubProb{i_{\beta}}\circ K)(x)$, where $i_{\beta}: y\mapsto y$ is the
injection, which yields a measurable map $(Y, \Borel{Y})\to (Y,
\Inv{\beta}{Y})$. Putting $X_{\alpha} := (X,
\Inv{\alpha}{X})$, similarly, for $Y_{\beta}$, this is equivalent to saying that this diagram
\begin{equation*}
\xymatrix{
X\ar[d]_{K}\ar[rr]^{i\alpha} && X_{\alpha}\ar[d]^{K'}\\
\SubProb{Y}\ar[rr]_{\SubProb{i_{\beta}}} && \SubProb{Y_{\beta}}
}
\end{equation*}
commutes. This is the diagram for stochastic relations which corresponds to
diagram~(\ref{lab:subsyst}). Thus the subsystems for stochastic
relations are exactly the congruences (and hence, for relations of the
form $K: S\Trans S$, the event bisimulations)\MMP{end new.}. 

Returning to the main stream of our discussion, we observe that surjective morphisms provide a rich source of
subsystems for the case of standard Borel spaces.
Before we can state and prove this, we need this auxiliary technical statement.

\BeginLemma{inv-under-i}
Let $X$ and $Y$ be standard Borel and $f: X\to Y$ be
measurable and surjective. Then $\InvBild{(\SubProbSenza i_{\Xf})}{D}$ is
$\SubProbSenza f$-invariant for each measurable subset $D$ of
$\SubProb{X_{\Xf}}$. 
\EndLemma

\BeginProof
We look at all sets which satisfy the assertion and show that these
sets comprise all measurable subsets $D$ of
$\SubProb{X_{\Xf}}$.  

In fact, let 
\begin{equation*}
  \mathcal{H} := \{D\subseteq \SubProb{\ul{Xf}}\mid \InvBild{(\SubProbSenza i_{\Xf})}{D}\text{ is $\SubProbSenza
    f$-invariant}\},
\end{equation*}
then this is a $\sigma$-algebra (see the proof of~\cite[Lemma
3.1.6]{Srivastava}), hence it is enough to show that $
\basS{\Sigma_{f}}{A}{q}\in\mathcal{H} $ whenever
$A\in\Sigma_{f}$. Then it will follow that the $\sigma$-algebra
generated by these sets is contained in ${\cal H}$, and these are all
measurable subsets of $\SubProb{X_{\Xf}}$.

Now assume that $(\SubProbSenza i_{\Xf})(\mu)\in
\basS{\Sigma_{f}}{A}{q}$, and assume that $\langle \mu, \mu'\rangle
\in\Kern{\SubProbSenza f}$, hence that  $(\SubProbSenza f)(\mu) =
(\SubProbSenza f)(\mu')$. Because $\InvBild{(\SubProbSenza
  i_{\Xf})}{\basS{\Sigma_{f}}{A}{q}}$ equals
$\basS{X}{A}{q}$, and because $A = \InvBild{f}{G}$ for some
Borel set $G\subseteq Y$ by Corollary~\ref{is-borel-conseq}, we infer
that
\begin{equation*}
  \mu'(A) = \mu'(\InvBild{f}{G}) = (\SubProbSenza f)(\mu')(G) =
  (\SubProbSenza f)(\mu)(G)
= \mu(A) \geq q,
\end{equation*}
hence $(\SubProbSenza
i_{\Xf})(\mu')\in\basS{\Sigma_{f}}{A}{q}$. Thus 
$
\basS{\Sigma_{f}}{A}{q}\in\mathcal{H},
$ 
and we are done.
\EndProof

This permit us to show that surjective morphisms define in fact
subsystems on their domains, i.e., on the effectivity functions which
serve as their source.

\BeginProposition{cont-def-subsystem}
Let $X$ and $Y$ be standard Borel spaces and $P: X\eTrans X$ resp.\ $Q: Y
\eTrans Y$ be stochastic
effectivity functions. A surjective morphism $f: P\to Q$ defines a
subsystem $\Sigma_{f}$ of $P$. In particular, defining 
\begin{equation*}
  P_{f}(x) := \{A\in w(\Xf) \mid
  \InvBild{\SubProb{i_{\Xf}}}{A}\in P(x)\}.
\end{equation*}
yields an effectivity function $P_{f}: X\eTrans (X, \Sigma_{f})$, and
$(id_{X}, id_{\Sigma_{f}}): P\to P_{f}$ is a morphism. 
\EndProposition

\BeginProof
1.  Given $H\in w(\Sigma_{f})\otimes[0, 1]$, we have to show
that the set 
\begin{equation*}
T_{H} := \{\langle x, t\rangle \mid
\InvBild{(\SubProbSenza i_{\Sigma_{f}})}{H_{t}}\in P(x)\}
\end{equation*}
is a member of
$\Sigma_{f}\otimes\Borel{[0, 1]}$. Since $P$ is an effectivity function, we
know that $T_{H}$ is a Borel set in $X\times[0, 1]$, so by
Lemma~\ref{is-borel-conseq-lem} it is enough to show that $T_{H}$
is $(f\times id_{[0, 1]})$-invariant.

2.  Given $t\in[0, 1]$, we know that $\InvBild{(\SubProbSenza
  i_{\Sigma_{f}})}{H_{t}}\in w(\Borel{X})$, by construction, and this
set is $\SubProbSenza f$-invariant by Lemma~\ref{inv-under-i}. Now for
showing the invariance property of $T_{H}$, we take $\langle x,
t\rangle \in T_{H}$ with $f(x) = f(x')$. Since
$\InvBild{(\SubProbSenza i_{\Sigma_{f}})}{H_{t}}\in P(x)$, and since this set
is $\SubProbSenza(f)$-invariant, we know that we can write $
\InvBild{(\SubProbSenza i_{\Sigma_{f}})}{H_{t}} = \InvBild{(\SubProbSenza
  f)}{G} $ for some $G\in \Borel{Y}$. Hence
\begin{align*}
 \InvBild{(\SubProbSenza i_{\Sigma_{f}})}{H_{t}} \in P(x) 
& \Leftrightarrow
 \InvBild{(\SubProbSenza f)}{G} \in P(x) \\
& \Leftrightarrow
G\in Q(f(x))&& \text{ since $f: P\to Q$ is a morphism}\\
& \Leftrightarrow
G\in Q(f(x'))&& \text{ since $f(x) = f(x')$}\\
& \Leftrightarrow
 \InvBild{(\SubProbSenza i_{\Sigma_{f}})}{H_{t}} \in P(x')
\end{align*}
Thus $\langle x', t\rangle\in T_{H}$.
\EndProof

Again\MMP{This is new.}, the situation is compared to stochastic relations and, by
implication, to their event bisimulations. Congruences and kernels of
surjective morphisms are closely related in this case, and the
discussion above has shown that congruences are nothing but subsystems
in disguise (or vice versa). This means that surjective morphisms are
in any case a resource from which we may harvest subsystems. Of
course, in the context of NLMPs and
effectivity functions, this demands further investigations into the
way subsystems integrate into this coalgebraic context. 

\subsection{Cospans For Finitely Generated Effectivity Functions}
\label{sec:cosp-finit-gener}

We will define behavioral equivalence and coalgebraic bisimulations
now and investigate their relationship for finitely generated
effectivity functions.

\BeginDefinition{beh-equiv-eff}
Let $P: S\eTrans T$ and $Q: X\eTrans Y$ be stochastic effectivity
functions.
\begin{itemize}
\item Call $P$ and $Q$ \emph{behaviorally equivalent} iff there exists
  a mediating effectivity function $M$ and surjective morphisms $
P \stackrel{(f, g)}{\longrightarrow} M \stackrel{(j, \ell)}{\longleftarrow} Q.
$
\item Call $P$ and $Q$ \emph{bisimilar} iff there exists a  mediating effectivity function $M$ and morphisms $
P \stackrel{(f, g)}{\longleftarrow} M \stackrel{(j, \ell)}{\longrightarrow} Q.
$
\end{itemize}
\EndDefinition
Hence behavioral equivalence is given through a co-span of morphisms,
bisimilarity by a span, as tradition demands. We will focus now on behavioral equivalent
effectivity functions $P$ and $Q$ for which the respective domains and
ranges are identical. It will be shown that we can find for finitely
supported effectivity functions bisimilar, albeit closely related functions,
which are defined on a subsystem. Finitely supported functions live on
a finite support set, which we show to be given by a countable set
of stochastic relations; this can be considered to be a version of the
notion of \emph{compactly generated} for effectivity functions.

We assume in this subsection that all spaces are standard Borel.

\paragraph{Preliminary Considerations.}
\label{sec:prel-cons}

Let $f: S\to U$ and $g: T\to U$ be both measurable and surjective
maps, and consider 
\begin{equation}
\label{def:w}
W := \{\langle s, t\rangle \mid f(s) = g(t)\} = \InvBild{(f\times g)}{\Delta_{U}}
\end{equation}
with $\Delta_{U}:= \{\langle u, u\rangle \mid u\in U\}$ as the
diagonal of $U$; intuitively, $\Delta_{U}$ is a true copy of $U$,
slightly tilted. Because $\Delta_{U}$ is closed in the underlying
topological space, $W$ is
measurable. Under the projections $\pi_{1}, \pi_{2}: \langle u, u\rangle
\mapsto u$, with inverse $d$, $\Delta_{U}$ is homeomorphic to $U$, so that
$\Borel{\Delta_{U}} = \InvBild{d}{\Borel{U}}$, consequently,
$w(\Delta_{U}) =
\InvBild{(\SubProbSenza{d})}{\Borel{\SubProb{U}}}$. We
have also
\begin{align*}
(f\circ \pi_{1, S})(s, t) & = (\pi_{1, U}\circ f\times g)(s, t),\\
(g\circ \pi_{2, T})(s, t) & = (\pi_{2, U}\circ f\times g)(s, t)
\end{align*}
for $\langle s, t\rangle\in S\times T$.  

Endow $W$ with the trace of the $\sigma$-algebra
$\Sigma_{f}\otimes\Sigma_{g}$, hence 
\begin{align*}
  \sigalg(W) & = \InvBild{(f\times g)}{\Borel{S\times S}}\cap
  \InvBild{(f\times g)}{\Delta_{U}}\\
& = \InvBild{(f\times g)}{\sigalg(\Delta_{U})}
\end{align*}

\paragraph{Finitely Supported Functions.}
\label{sec:finit-supp-funct}

A finitely supported effectivity function on $S$ can be represented through stochastic relations, which are obtained as measurable selections.
\BeginProposition{exists-stoch-rel}
Let $P:S \eTrans S$ be finitely supported by $E$, then there exists a countable set
$\mathcal{K}$ of stochastic relations $S\Trans S$ such that $E(s) =
\mathcal{K}(s) := \{K(s) \mid K\in \mathcal{K}\}$. 
\EndProposition

\BeginProof
Since $S$ is standard Borel, we find a Polish topology which generates
the $\sigma$-algebra; assume that $S$ is endowed with this topology,
then $\SubProb{S}$ is a Polish space as well under the weak
topology. We take this topology. 
Let $G\subseteq \SubProb{S}$ be open, then 

\begin{equation*}
  \{s\in S \mid E(s)\cap G \not= \emptyset\} = \bigcup_{e\in E}\{s\in
  S\mid e(s)\in G\}\in \Borel{S}.
\end{equation*}
Thus there exists a Castaing representation
$\Folge{K}$ for $E$ by the Selection Theorem~\ref{Himmelberg}. Because $E$ is finite, we have 
$
E(s) = \{K_{n}(s) \mid n\in \Nat\}.
$
\EndProof

Note that although $\Folge{K}$ is a countable sequence of stochastic
relations, the representation $E(s) = \{K_{n}(s) \mid n\in \Nat\}$
together with the finiteness of $E(s)$ implies that the set
$\{K_{n}(s) \mid n\in \Nat\}$ is finite for each state $s$, but that we
cannot conclude that $\Folge{K}$ is a finite sequence.

Now assume that the finitely supported effectivity functions $P$ and
$Q$ are behaviorally equivalent, so that we have a cospan
$P\stackrel{f}{\rightarrow} M \stackrel{g}{\leftarrow} Q$
and $f$ and $g$ surjective
morphisms. The co-span expands to 
\begin{equation*}
\xymatrix{
S\ar[rr]^{f}\ar[d]_{P} && U\ar[d]_{M} && T\ar[d]^{Q}\ar[ll]_{g}\\
\Vau{S}\ar[rr]_{\VauSenza f} && \Vau{U} && \Vau{T}\ar[ll]^{\VauSenza g}
}
\end{equation*}
Let $\mathcal{K}$ and $\mathcal{L}$ be the countable sets of
stochastic relations associated with $P$ resp. $Q$ according to
Proposition~\ref{exists-stoch-rel}. We will show now that $M$ is
finitely supported, and that the pointwise images of ${\cal K}$ and ${\cal L}$
under $\SubProbSenza{f}$ resp. $\SubProbSenza{g}$ coincide whenever
the image of $f$ and $g$ are the same. 

\BeginProposition{also-finitely-supported}
In the notation above, $M$ is finitely supported, and $\Bild{(\SubProbSenza f)}{\mathcal{K}(s)}=\Bild{(\SubProbSenza g)}{\mathcal{L}(t)}$, whenever $f(s) = g(t)$.
\EndProposition

\BeginProof
Because $f: P\to M$ is a morphism, we have $ G\in M(f(s)) $ iff $
\InvBild{(\SubProbSenza f)}{G}\in P(s), $ which in turn is equivalent
to $ \mathcal{K}(s) \subseteq\InvBild{(\SubProbSenza f)}{G}.  $ Hence
$ \Bild{(\SubProbSenza f)}{\mathcal{K}(s)}\in M(f(s)).  $ Similarly,
we conclude $ \Bild{(\SubProbSenza g)}{\mathcal{L}(t)}\in M(g(t)).  $
Because $f$ and $g$ both are onto, we find for each $u\in U$ some
$s\in S$ and $t\in T$ with $f(s) = u = g(t)$. From this we conclude
that $M(u)$ is supported both by $\Bild{(\SubProbSenza
  f)}{\mathcal{K}(s)}$ and by $\Bild{(\SubProbSenza
  g)}{\mathcal{L}(t)}$.
\EndProof

Before defining an effectivity function, we briefly investigate the
$\SubProbSenza$-image of the projections $\pi_{S}$ and $\pi_{T}$. We
show that they induce surjective maps, so that we can be sure that
each element of the corresponding range occurs as an image. This is of
technical use when investigating the system dynamics for the mediator.

\BeginLemma{gen-3}
$\SubProb{\pi_{S}}: \SubProbSenza(W, \Sigma_{f\times g}\cap W) \to
\SubProbSenza(S, \Sigma_{f})$ and
$\SubProb{\pi_{T}}: \SubProbSenza(W, \Sigma_{f\times g}\cap W) \to
\SubProbSenza(T, \Sigma_{f})$
are both onto.
\EndLemma

\BeginProof
We have a look at $\pi_{S}$ only, the argumentation is symmetric for
its step twin $\pi_{T}$. 
Define $\psi: \sigalg(Y)\to \Sigma_{f}$ through $D\mapsto
\InvBild{f}{D}$, then $\psi$ induces a bijection $\Psi:
\SubProb{Y}\to \SubProb{S, \Sigma_{f}}$ by
Lemma~\ref{gen-1}. Given $\mu\in\SubProb{S, \Sigma_{f}}$, define $$\nu
:= (\SubProbSenza f)(\mu),$$ hence $\nu\in\SubProb{U}$. Because
$e_{1}\circ f\times g: W\to U$ is onto, we find $\nu'\in\SubProb{W,
  \Sigma_{f\times g}\cap W}$ with $$\nu = \SubProb{e_{1}\circ f\times
  g}(\nu').$$ But $$e_{1}\circ (f\times g) = f\circ \pi_{S}$$ on $W$, so
that $$\nu = (\SubProb{f}\circ \SubProb{\pi_{S}})(\nu') =
(\SubProbSenza 
f)(\mu).$$ Hence $$
\mu = \Psi\bigl(\SubProb{f}\circ
\SubProb{\pi_{S}}\bigr)(\nu') = (\SubProbSenza \pi_{S})(\nu').
$$
\EndProof

\paragraph{Constructing the Mediator.}
\label{sec:constr-medi}

Now define $W$ as in (\ref{def:w}) and endow it with
$\InvBild{(f\times g)}{\Borel{\Delta_U}}$ as a $\sigma$-algebra. We
know from Proposition~\ref{also-finitely-supported} that $\InvBild{(\SubProbSenza \pi_S)}{\mathcal{K}(s)}\subseteq Z $
iff $\InvBild{(\SubProbSenza
  \pi_T)}{\mathcal{L}(t)}\subseteq Z$ for $f(s) = g(t)$. This suggests
the following definition for the dynamics for $\langle s, t\rangle\in W$.
\begin{equation}
\label{def:tau}
\tau(s, t) := \fil{\bigl(\InvBild{(\SubProbSenza
    \pi_S)}{\mathcal{K}(s)}\bigr)}
\big(= \fil{\bigl(\InvBild{(\SubProbSenza \pi_S)}{\mathcal{L}(t)}\bigr)}\big)
\end{equation}
Hence $\tau: W\to \Vau{W}$ renders this diagram commutative:
\begin{equation*}
\xymatrix{
S\ar[d]_{P_f} && W\ar[d]_\tau\ar[rr]^{\pi_S}\ar[ll]_{\pi_T} && T\ar[d]^{Q_g}\\
\Vau{S_f} && \Vau{W}\ar[rr]_{\VauSenza \pi_S}\ar[ll]^{\VauSenza \pi_T} && \Vau{T_g}
}
\end{equation*}
In fact, if $\langle s, t\rangle\in W$, we have 
\begin{align*}
\InvBild{(\SubProbSenza \pi_S)}{Z}\in \tau(s, t)
& \Leftrightarrow
\InvBild{(\SubProbSenza \pi_S)}{\mathcal{K}(s)}\subseteq\InvBild{(\SubProbSenza \pi_S)}{Z}\\
& \Leftrightarrow
\mathcal{K}(s)\subseteq Z
\end{align*}
The last conclusion follows from the observation that since $\pi_S$ is
onto and continuous, $\SubProbSenza \pi_S$ is onto as well, so that
$(\SubProbSenza \pi_S)^{-1}$ is injective as a set valued map.

Hence we see in a similar way, for $\langle s, t\rangle\in W$
\begin{align*}
P_f(s) & = \{ Z\in \Borel{\SubProbSenza W}| \InvBild{(\SubProbSenza \pi_S)}{Z}\in\tau(s, t)\},\\
Q_g(t) & = \{ Z\in \Borel{\SubProbSenza W}| \InvBild{(\SubProbSenza \pi_S)}{Z}\in\tau(s, t)\}.
\end{align*}

\BeginProposition{tau-is-t-measurable}
$\tau$ is t-measurable.
\EndProposition

Before we are in a position to establish Proposition~\ref{tau-is-t-measurable}, we need an auxiliary statement.
\BeginLemma{l-1}
Define $e_i: \Delta_U \to U$ as the $i$-th projection. If $H\in w(W)$, then there exists $H'\in w(U)$ such that $H = \InvBild{\bigl(\SubProbSenza (e_1\circ f\times g)\bigr)}{H'}$.
\EndLemma

\BeginProof
Let $ \mathcal{H} $ be the set of all $H\in w(W)$ for
which the assertion is true. Then $\mathcal{H}$ is a $\sigma$-algebra,
and $\basS{W}{D}{\bowtie q}\in\mathcal{H}$, whenever
$D\in\Sigma_{f\times g}\cap W$. In fact, $D$ can be written as $D =
\InvBild{(f\times g)}{D_1}$ for some $D_1\in\sigalg(\Delta_U)$, and
$D_1 = \InvBild{e_1}{D_0}$ for some $D_0\in\Borel{U}$. Hence $D =
\InvBild{(e_1\circ f\times g)}{D_0}$, from which\begin{equation*}
\basS{W}{D}{\bowtie q} = \InvBild{\bigl(\SubProbSenza (e_1\circ f\times g)\bigr)}{\beta_U(D_0, \bowtie q)}
\end{equation*}
follows. Thus $\mathcal{H}$ contains the generator for $w(W)$, from which the assertion follows.
\EndProof

This yields as an immediate consequence the following observation.

\BeginCorollary{c-l-1}
If $H\in \sigalg\bigl(\SubProb{W}\otimes[0, 1]\bigr)$, then there exists $H'\in \sigalg\bigl(\SubProb{U}\otimes[0, 1]\bigr)$ such that $H = \InvBild{\bigl(\SubProbSenza (e_1\circ f\times g)\times id_{[0, 1]}\bigr)}{H'}$.
\EndCorollary

This Corollary will permit this proof strategy: if we want to establish a
property for some measurable $H\subseteq \SubProb{W}\otimes[0, 1]$, we
investigate $H'$ instead, establish a suitably modified property for
$H'$ and prove things for $H'$. But let us have a look at the proof.

\BeginProof
Again, the set for which the assertion is true is a $\sigma$-algebra,
which contains by Lemma~\ref{l-1} all measurable rectangles $H_1\times
V$ with $H_1\in w(W)$ and $V\in\Borel{[0, 1]}$.
\EndProof

\BeginProof (of Proposition~\ref{tau-is-t-measurable})
1.
Since $\tau(s, t) = \fil{\bigl(\InvBild{(\SubProbSenza \pi_S)}{\mathcal{K}(s)}\bigr)}$ for $\langle s, t\rangle\in W$, we have to show that 
\begin{equation*}
T_H := \{\langle s, t, q\rangle \mid  H_q\in \tau(s, t)\}
\end{equation*}
is a member of $\sigalg(W\otimes[0, 1])$, whenever $H\subseteq \SubProb{W}\times[0, 1]$ is measurable. Thus we have to show that  
$
T_H = \{\langle s, t, q\rangle\in W\times[0, 1] |\InvBild{(\SubProbSenza \pi_S)}{\mathcal{K}(s)}\subseteq H_q\}\in\Borel{W\otimes[0, 1]}.
$

2.  
Now Corollary~\ref{c-l-1} kicks in. For $H$ there exists $G\in\sigalg\bigl(\SubProb{U}\otimes[0, 1]\bigr)$ such
that $H_q = \InvBild{\bigl(\SubProbSenza(e_1\circ f\times
  g)\bigr)}{G_q}$. 

We claim 
\begin{equation*}
\InvBild{(\SubProbSenza \pi_S)}{\mathcal{K}(s)}\subseteq H_q
\Leftrightarrow
\mathcal{K}(s)\subseteq 
\InvBild{(\SubProb{f}\times id_{[0, 1]})}{G}_q  
\end{equation*}

``$\Leftarrow$'': 
Note that $f\circ \pi_{S} = e_{1}\circ f\times g$ on $W$, hence
\begin{align*}
  \InvBild{(\SubProbSenza \pi_{S})}{\mathcal{K}(s)} 
& \subseteq
\InvBild{(\SubProbSenza \pi_{S})}{\InvBild{(\SubProb{f}\times id_{[0, 1]})}{G}_q } \\
& =
  \InvBild{(\SubProbSenza \pi_{S})}{\InvBild{(\SubProbSenza
      f)}{G_{q}}}\\
& = 
\InvBild{(\SubProb{f\circ \pi_{S}})}{G_{q}}\\
& =
\InvBild{(\SubProb{e_{1}\circ f\times g})}{G_{q}}\\
& =
H_{q}
\end{align*}
``$\Rightarrow$'': 
Because $\SubProbSenza \pi_{s}$ is onto by Lemma~\ref{gen-3}, we infer
from $\InvBild{(\SubProbSenza \pi_S)}{\mathcal{K}(s)}\subseteq H_q$
that $\mathcal{K}(s)\subseteq \Bild{(\SubProbSenza \pi_{S})}{H_{q}}$. 
Now let $\mu\in \Bild{(\SubProbSenza \pi_{S})}{H_{q}} = \Bild{(\SubProbSenza \pi_S)}{\InvBild{\bigl(\SubProbSenza(e_1\circ f\times
    g)\bigr)}{G_q}}$, then there exists $\nu\in
\InvBild{\bigl(\SubProbSenza(e_1\circ f\times g)\bigr)}{G_q}$ with
$\mu = (\SubProbSenza \pi_S)(\nu)$, hence
\begin{equation*}
(\SubProbSenza f)(\mu) = (\SubProbSenza f\circ \pi_S)(\mu) = (\SubProbSenza(e_1\circ f\times g))(\nu) \in G_q.
\end{equation*}
3.
Hence we have 
$
T_H = \{\langle s, t, q\rangle \in W\times [0, 1]\mid \mathcal{K}(s) \subseteq \InvBild{(\SubProb{f}\times id_{[0, 1]})}{G}_q\},
$
which is a measurable subset of $W\times[0, 1]$. 
\EndProof

Concluding, we have shown

\BeginProposition{bisim-exists}
Let $P: S\eTrans S$ and $Q: T\eTrans T$ be finitely supported and
behaviorally equivalent stochastic effectivity functions,
and assume that $S$, $T$ and the mediator's state space are standard
Borel spaces. Then there exist subsystems ${\cal C}_{P}$ and ${\cal
  C}_{Q}$ of $P$ resp.\ $Q$ and effectivity functions $P_{s}: S\eTrans
(S, {\cal C}_{P})$ and $Q_{s}: T\eTrans (T, {\cal C}_{Q})$ such that
$P_{s}$ and $Q_{s}$ are bisimilar. \qed
\EndProposition

Let us briefly look back and see what we did, and how we did it. We
identified subsystems for the given morphisms. On the equalizer $W$ of
these morphisms, we constructed a Borel structure, and from this
measurable space we obtained a stochastic effectivity function. So far
the proof resembles the classic proof for the existence of
bisimulations for set based systems. The complications arise when
having to establish that we have here a span of morphisms, in part
because the measurable structure on the subprobabilities of $W$ is
given by appealing to the measurable structure on the diagonal of the
target space for the co-span, from which we started, through the given
morphisms and through projections. Investigating this structure, which
arises through these delegations (as an object oriented programmer
would say) and which requires the maps and their
images under $\SubProbSenza$ as handles, renders the proof somewhat
involved.

One might be tempted now to capitalize on this result for establishing the existence
of a bisimulation for stochastic relations. This approach would rest
on the observation that each stochastic relation $K: S\Trans T$ for
measurable spaces $S$ and $T$ generates an effectivity function
$P_{K}$ upon setting $P_{K}(s) := \{A\in w(S)\mid
K(s)\in A\}$. Hence $P_{K}$ is clearly finitely supported. Since each
morphism for stochastic relations yields a morphism for the associated
effectivity function, the result above applies to stochastic relations
$K: S\Trans S$ and $L: T\Trans T$, and one finds a mediating
effectivity function for $P_{K, f}$ and $Q_{L,
  g}$. This effectivity function is finitely supported, but not by the
associated Borel sets. It is rather defined on a sub-$\sigma$-algebra
of the Borel sets of $W$, the space we have identified here. It
cannot be concluded right away, however, that the stochastic relation which
underlies the mediating function is also a stochastic relation for the
Borel sets of the underlying Polish or analytic spaces; looking at the
proofs establishing the existence of a semi-pullback~\cite{Edalat,EED-CongBisim}, it turns out
that it is this property which is the crucial one for stochastic relations. 

\subsection{Event Bisimulations Revisited}
\label{sec:event-bisimulations}
Since both event bisimulations and subsystems are given by
sub-$\sigma$-algebras, it is natural to  inquire about the connections
between these two notions, which adapt the remarks made above
concerning the relationships between subsystems and event
bisimulations for stochastic relations. 

In first place, we need to make some detailed observations concerning
the definition of event bisimulations. Fix a nondeterministic kernel $\ka :  S\to \weak(S)$. We have
defined an event bisimulation to be any sub-$\sigma$-algebra $\calC$ of
$\sigalg(S)$ such that
$\ka:(S,\calC)\to(\weak(S),\hit(\weak(\calC)))$ is measurable. This
means that for all $G\in \weak(\calC)$, the preimage
\begin{equation*}
\InvBild{\ka}{\hitt_G} = \{s\in S \st \ka(s) \cap G \neq \emptyset\}
\end{equation*}
should belong to $\calC$. But we have to exercise some care here:
although any measure $\mu$ in $\subp(S) =\subp(S,\sigalg(S))$ assigns
also values to elements of ${\cal C}$, it is actually not a member of
the space $\subp(S,\calC)$. We should restrict its domain by using the
map $\subp \iotac:\subp(S,\sigalg(S))\to\subp(S,\calC)$. So we finally
obtain the following more precise formulation of event bisimulation
on  $\ka :    S\to \weak(S)$:
any $\calC$ such that for all $G\in \weak(\calC)$,
\begin{equation*}
\ka^{-1}[\hitt_{(\subp \iotac)^{-1}[G]}] = \{s\in S \st \ka(s)
\cap{(\subp \iotac)^{-1}[G]} \neq \emptyset\} \in{\cal C}.
\end{equation*}
Finally, by considering  Lemma~\ref{l:H-meas_subseteq}, this is
equivalent to having
\begin{equation}\label{eq:12}
  \{s\in S \st \ka(s)\subseteq(\subp\iotac)^{-1}[G]\}\in\calC
\end{equation}
for all $G\in\weak(\calC)$.

As for stochastic relations, subsystems induce event
bisimulations. This is fairly straightforward.

\BeginProposition{is-event-bisim}
  Let $\calC$ be a subsystem of $\fil{\ka}:S\to\weak(S)$. Then $\calC$ is
  an event bisimulation of $\ka$.  
\EndProposition
\begin{proof}
  The hypothesis tells us that for all
  $H\in\sigalg(\subp(S,\calC)\otimes [0,1])$, we have
  \[    \{\<s,q\> \st H_q\in (\fil{\ka})_\calC(s)\} \in
  \sigalg((S,\calC)\otimes [0,1]) \]
  that is,
  \begin{equation}\label{eq:11}
    \{\<s,q\> \st (\subp\iotac)^{-1}[H_q]\in \fil{\ka(s)}\} \in
    \sigalg((S,\calC)\otimes [0,1])
  \end{equation}

  We aim at showing that (\ref{eq:12}) holds for all $G\in\weak(\calC)$.

  In order to do this, take $H:=G\times[0,1]$ in (\ref{eq:11}) and
  calculate:
  \begin{align*}
    \{\<s,q\> \st (\subp\iotac)^{-1}[H_q]\in \fil{\ka(s)}\} &= \{\<s,q\>
    \st (\subp\iotac)^{-1}[G]\in \fil{\ka(s)}\}\\
    &= \{\<s,q\>
    \st \ka(s)\subseteq (\subp\iotac)^{-1}[G]\}\\
    &= \{s \st \ka(s)\subseteq (\subp\iotac)^{-1}[G]\} \times [0,1]    
  \end{align*}
  Now the set $\{s \st \ka(s)\subseteq (\subp\iotac)^{-1}[G]\}$ is a
  cut of a set in $\sigalg((S,\calC)\otimes [0,1]) =
  \calC\otimes\B([0,1])$, and hence it belongs to $\calC$.
\end{proof}

We will need for the converse to strengthen our hypothesis by
following the ideas put forward in
Proposition~\ref{pr:parametrized-hit}, making the parameter from $[0,
1]$ explicit. It works like this: 
Let $\ka:S\to\weak(S)$ be a generating
kernel. We will show that parametrized  event
bisimulations, appropriately defined for the extended system with base
space $S\times [0,1]$, are exactly the subsystems of the effectivity function
based on $\ka$. 

Formally, a \emph{parametrized event  bisimulation on
  $\ka:S\to\weak(S)$}  is a sub-\sig-algebra of $\sigalg(S)$ such that 
$\ka\times id_{[0, 1]}:S\times [0,1]\to \weak(S)\times
[0,1]$ is
$(\calC\otimes\B([0,1]))$-$\hit(\weak(\calC)\otimes\B([0,1]))$ measurable: for all
$G\in\weak(S,\calC)\otimes\B([0,1])$, 
\begin{equation*}
  \{\<s,q\>\in S \times [0,1]\st (\ka\times id_{[0,1]})(s,q)\subseteq  (\subp\iotac\times
  id_{[0,1]})^{-1}[G]\} \in \calC\otimes\B([0,1]),
\end{equation*}
compare this with (\ref{eq:12}).
\begin{prop}
  For every generating
  kernel $\ka:S\to\weak(S)$, the parametrized event  bisimulations on
  $\ka$ are exactly the subsystems of $\fil{\ka}$.
\end{prop}
\begin{proof}
  Let $\calC$ be a parametrized event bisimulation. By recalling
  (\ref{eq:11}) in the proof of the previous proposition, we would
  like to show that
  \begin{equation*}
    \{\<s,q\> \st (\subp\iotac)^{-1}[G_q]\in \fil{\ka(s)}\} \in
    \sigalg((S,\calC)\otimes [0,1]).
  \end{equation*}
  for all $G\in\sigalg(\subp(S,\calC)\otimes [0,1])$. By definition of $\fil{\cdot}$, we may write this set as
  \[  \{\<s,q\> \st \ka(s)\subseteq (\subp\iotac)^{-1}[G_q]\}.\]
  Next make explicit the cut component $q$:
  \[ \{\<s,q\> \st \ka(s)\times \{q\}\subseteq  (\subp\iotac\times
  id_{[0,1]})^{-1}[G]\}.\]
  That is, we need
  \[ \{\<s,q\> \st (\ka\times id_{[0,1]})(s,q)\subseteq  (\subp\iotac\times
  id_{[0,1]})^{-1}[G]\} \in \calC\otimes\B([0,1]).\]
  But this last condition is exactly the definition of parametrized
  event bisimulation, and therefore the two notions are equivalent.
\end{proof}
To finish this section, we will show that in a finitary setting, the relational approach to
bisimilarity of effectivity functions is compatible to that of
subsystems, by showing that the 
relation of state bisimilarity induces a subsystem in a natural
way. We actually obtain a stronger result: every smooth bisimulation
induces a subsystem, the result on finitary effectivity functions
follows as a corollary. 

A technical lemma is required first.
\BeginLemma{lem:bisim-subsys}
  Let $P:S\eTrans S$ an effectivity function, $R$ a state bisimulation on $P$, and
  $\calC:=\Sigma_{R}$. 
  Then  for given
  $H\in\sigalg(\subp(S,\calC)\otimes[0,1])$, 
assume $H_{q}\in P_\calC(s)$, and let $\isEquiv{s}{t}{R}$. Then $H_{q}\in
P_\calC(t)$. 
\EndLemma

\begin{proof}
  Since $R$ is a bisimulation, if $s\mathrel{R} t$ we have
  \begin{equation*}
    \forall G\in P(s)\,\exists E\in P(t):\forall \nu \in E \,\exists
    \mu \in G\; (\mu \mathrel{\bar R} \nu).
  \end{equation*}
  By definition of $\iotac$,
  \begin{equation*}
    \forall G\in P(s)\,\exists E\in
    P(t):\subp\iotac[E]\subseteq\subp\iotac[G].
\end{equation*}
  Now we can take $E':=E\cup G$, and this set belongs to $P(t)$ by
  upper-closedness, and we obtain
  \begin{equation}\label{eq:8}
    \forall G\in P(s)\,\exists E'\in
    P(t):\subp\iotac[E']=\subp\iotac[G].
  \end{equation}
 This yields
  \begin{align*}
    A &:= \{\<s,q\> \st H_q\in P_\calC(s)\} \\ 
      &= \{\<s,q\> \st (\subp\iotac)^{-1}[H_q]\in P(s)\} &&
    \text{ definition of }P_\calC\\ 
      &= \{\<s,q\> \st \exists D\in P(s) :
    (\subp\iotac)^{-1}[D]\subseteq H_q\in P(s)\} 
  \end{align*}
  If  $\<s,q\>\in A$  there exists a $D$ as in
  the previous line. And if $ s\mathrel{R} t$, there must exist some
  $D'\in P(t)$ such that $\subp\iotac[D']=\subp\iotac[D]$ by (\ref{eq:8}). Hence
  $\<t,q\>\in A$.
\end{proof}

Now all preparations are finally in place for showing that bisimulation
equivalences induce subsystems, provided they are smooth. By using our
logical characterization for finitary effectivity functions we may also
prove that state bisimilarity induces a subsystem.
\begin{theorem}\label{t:smooth-bisim-subsys}
  Let $S$ be Polish,  $P:S\to\cones(S)$ an effectivity function, and $R$ a smooth state
  bisimulation equivalence on $P$. Then 
  $\Sigma_{R}$ is a  subsystem. In particular, if $P$ is finitary,
  then
  $\Sigma({\sbisim},\sigalg(S))$ is a  subsystem.
\end{theorem}

\begin{proof}
1.
  We have to prove that $P_\calC$ is t-measurable. That means, for all
  $H\in\sigalg(\subp(S,\calC)\otimes[0,1])$, the set $A$ as defined
  for $H$ in the
  proof for Lemma~\ref{lem:bisim-subsys}
should belong to $\Inv{R}{\B(S)}\otimes\B([0,1])$.
  By Lemma~\ref{lem:bisim-subsys}, we have
  $A\in\Inv{R\times \Delta_{[0,1]}}{\B(S)\otimes\B([0,1])}$. But since
  $[0, 1]$ is a Polish space, and since $R$ is smooth, we obtain from~\cite[Lemma
  6.4.2]{Bogachev} via~\cite[Proposition 2.12]{EED-Alg-Prop_effFncts} that
  \begin{equation*}
 \Inv{R}{\B(S)}\otimes\B([0,1])=\Inv{R\times \Delta_{[0,1]}}{\B(S)\otimes\B([0,1])},
\end{equation*}

2.
  By   Corollary~\ref{c:bisimilarity-smooth} the relation $\sbisim$ of
  state bisimilarity on $P$ is smooth, hence
  Theorem~\ref{t:smooth-bisim-subsys} applies.
\end{proof}

\section{Conclusions \& Further Work}
\label{sec:conclusions--further-work}
We shed some light into the relationship of nondeterministic kernels
and stochastic effectivity functions, forging links between the
respective concepts of morphism and bisimulation. While
nondeterministic kernels appear as a natural generalization
of 
purely probabilistic systems, the properties proposed for effectivity functions
are useful for interpreting game logics.
Both are descendants of stochastic relations, and this common origin
is one reason for having a tighter bond between the finitely-generated
versions of each of them. Two of the main contributions of this work
concern the finitary setting, namely: the fact that image-finite
kernels always give rise to effectivity functions via the filter
construction, and that state bisimilarity on finitary effectivity
functions is characterized by the two-level modal logic $\two$.

We also started the coalgebraic study of nondeterministic kernels, by
translating the notion of morphism for effectivity functions to them,
thus obtaining a category of NLMPs. There is still much work to be
done in this direction, since we still only have some preliminary
approximations for  the candidate  endofunctor for the coalgebraic
structure. One pertinent observation is that if we obtain a
successful definition of this functor, it would most probably be a
contravariant one. This is not yet encompassed in the standard
theory of coalgebras, so a systematic study of \emph{contravariant
  coalgebras} might be a natural step to follow.

In the opposite way, the study of the translation from
nondeterministic kernels to effectivity functions made apparent the
need of different notions of morphism for the latter. The plain
definition of effectivity function morphism corresponds to bounded
morphisms for models of monotonic non-normal modal logic, as studied
by Hansen \cite{hansen2003monotonic}. Now, both our finitary
effectivity functions and those generated by a kernel are analogous to
\emph{core-complete} models, meaning that their value for each state
is a union of principal filters. In this context, strong morphisms
correspond to \emph{bounded core morphisms}.  Somehow related to this,
it is important to note that the working hypothesis used to prove a
Hennessy-Milner theorem for monotonic models (i.e., a logical
characterization of bisimilarity~\cite[Definition
4.30]{hansen2003monotonic}) is the same as ours of having a finitary
effectivity function. Further connections between the standard theory
of monotonic models and the present stochastic version should be
studied in the future.

We would like to see an interpretation of game logic without the
fairly strong condition of t-measurability, which, however, has turned
out to be technically necessary. 
One point of attack may be given by the
observation that certain crucial real functions that appear in the
study of composability of effectivity functions \cite[Section~3.2]{EED-Alg-Prop_effFncts} are  monotonically
decreasing and hence Borel measurable.
If it could
be verified that t-measurability can be deduced from some more
lightweight properties of nondeterministic kernels, the boundaries
between NLMPs and effectivity functions could be lowered, permitting
an easier exchange of properties. It seems necessary to start by
investigating the composition of two nondeterministic kernels; there
are natural definitions of this composition, stemming from the
discrete case, but  measurability obstacles preclude a direct
generalization.

Another topic of interest is the
relationship of the two-level logic $\two$ and game logic. The
interpretation of the ``upper level'' of $\two$ in terms of measures would
open up a lot of interesting perspectives on game logic, not least for
the exploration of bisimulations for probabilistic neighborhood models, about we do
not know too much even under the assumption of Polish or analytic
state spaces. Also, we will give an account for the relation between
$\two$ and the logic for NLMPs somewhere else. This may help in completing the picture in
this stochastic landscape.

Finally  a third area of further work includes a
deeper study of coalgebraic matters, in particular of expressivity
including bisimilarity. We believe that (at least under the assumption
of finitary effectivity functions) stronger results may be obtained.

\begin{ack}
  Most of the work was done while the second author visited the first
  author's group at Technische Universit\"at Dortmund. The authors want
  to thank DFG for funding this stay and making this cooperation
  possible; the second author wants to thank TU Dortmund staff for
  making his research stay both fruitful and highly enjoyable. Both
  authors had a lot of nonMarkovian fun with their joint research.
\end{ack}


\end{document}